\documentclass[twocolumn,superscriptaddress,tightenlines,aps,prx,floatfix]{revtex4-2}

\usepackage{graphicx}
\usepackage{amsmath}   
\usepackage{amsfonts}
\usepackage{physics}
\usepackage{amssymb}
\usepackage{booktabs}
\usepackage{microtype}
\usepackage{amsthm}

\usepackage{xcolor}
\definecolor{DarkGreen}{RGB}{1,80,1}
\definecolor{Glicine}{RGB}{180, 130, 200}

\newtheorem{theorem}{Theorem}
\newtheorem{definition}{Definition}

\newtheorem{lemma}{Lemma}
\usepackage[colorlinks=true, allcolors=blue]{hyperref}

\usepackage{cleveref}
\usepackage{appendix}
\usepackage{dsfont}
\usepackage{mathtools}
\usepackage{caption}
\usepackage{mathrsfs}
\usepackage{titlesec}
\usepackage{float}
\usepackage{bibunits}
\defaultbibliographystyle{apsrev4-2}

\begin{document}
\title{Harmonic morphisms and dynamical invariants in network renormalization}

\author{Francesco Maria Guadagnuolo}
\affiliation{EPFL, Lausanne, Switzerland}

\author{Marco Nurisso}
\affiliation{Dipartimento di Scienze Matematiche, Politecnico di Torino, Torino, Italy}

\author{Federica Galluzzi}
\affiliation{Dipartimento di Matematica, Università di Torino, Torino, Italy}

\author{Antoine Allard}
\affiliation{D\'epartement de physique, de g\'enie physique et d'optique, Universit\'e Laval, Qu\'ebec, Canada}
\affiliation{Centre interdisciplinaire en mod\'{e}lisation math\'{e}matique, Universit\'{e} Laval, Qu\'{e}bec, Canada}

\author{Giovanni Petri}
\affiliation{Network Science Institute, Northeastern University London, London, UK}

\begin{bibunit}

\begin{abstract}
Renormalization of complex networks requires principled criteria for assessing whether a coarse-graining preserves dynamical content.
We prove that discrete harmonic morphisms ---surjective maps preserving harmonic functions--- provide the minimal condition under which random walks on a fine-grained network project exactly onto random walks on its coarse-grained image, through an appropriate random time change.
We formalize this via the \emph{harmonic degree}, a diagnostic quantifying how closely any network coarse-graining approximates a harmonic morphism.
Applying this framework to geometric, Laplacian, and GNN-based renormalization across real-world networks, we find that each method produces a distinct dynamical fingerprint encoding its underlying physical assumptions.
Most strikingly, Laplacian renormalization spontaneously yields \emph{exact} harmonic morphisms in several networks, achieving exact preservation of first-exit random-walk transition structure at specific scales, a property that entropic susceptibility fails to detect.
Our results identify a discrete analog of diffusion-preserving conformal maps for irregular network topologies and provide quantitative tools for designing and evaluating multi-scale network descriptions.
\end{abstract}

\maketitle

%% ====================================================================
%%  INTRODUCTION
%% ====================================================================
\section{Introduction}

The renormalization group (RG) is a fundamental tool in statistical mechanics for studying how physical systems change with observation scale~\cite{Kadanoff_renorm,Wilson_renorm_critical_phenomena_Kondo}. For regular lattices, conformal invariance provides clear criteria for scale-invariant descriptions~\cite{Conformal_Invariance_Lattice_Models_Smirnov,Introduction_to_Conformal_Field_Theory}, and renormalization transformations preserving these symmetries maintain essential physics while changing scale.
Extending the RG from lattices to complex networks remains an open challenge. Numerous approaches have been proposed~\cite{Song_self_similarity_complex_nets_2005,Skeleton_and_Fractal_Scaling, Multiscale_unfolding_of_real_nets_geometric_renorm_2018, Garlaschelli_Multiscale_2023, Villegas_Laplacian_renorm_heterogeneous_2023, Coarse_graining_network_De_Domenico}, including treatments for higher-order structures~\cite{nurisso2025higher, Path_integral_higher_order_renorm, Battiston2021physics}. 
Each method focuses on specific network characteristics: geometric structure through hyperbolic embeddings~\cite{Multiscale_unfolding_of_real_nets_geometric_renorm_2018}, diffusive structure via spectral properties~\cite{Villegas_Laplacian_renorm_heterogeneous_2023}, latent probabilistic laws~\cite{Garlaschelli_Multiscale_2023}, or informational content~\cite{Coarse_graining_network_De_Domenico}. Nonetheless fundamental questions persist: how should we evaluate a renormalization procedure? How do we assess whether coarse-graining preserves dynamics, and not merely structure? Despite recent progress~\cite{chen2025preservingspreadingdynamicsinformation,yi2025equilibriumpreservinglaplacianrenormalizationgroup,Schmidt2025Spectral,KIM2025117398}, dynamical renormalization of complex networks has not been systematically investigated.

Our approach differs from prior work in that we focus on the coarse-graining transformation itself, treating it as a surjective function between graphs and searching for its symmetries and preserved quantities. Specifically, we leverage the notion of \emph{discrete harmonic morphism}, introduced by Urakawa~\cite{Urakawa_Harmonic_morphisms_2000}. We prove that harmonic morphisms ---mappings preserving locally harmonic functions--- send random walks on fine-grained networks to random walks on coarse-grained networks through appropriate random time changes: a harmonic morphism equates first-passage probabilities for random walkers exiting a group of nodes (named ``macro-sets'' below) with one-step transition probabilities between corresponding macro-nodes. 
This parallels the role of conformal maps in 2D field theory, which preserve Brownian motion up to time reparametrization, and generalizes this diffusion-preserving property to discrete, irregular topologies via the natural framework of harmonic morphisms.

Building on this theoretical foundation, we introduce the \emph{harmonic degree}, a diagnostic quantifying how much any coarse-graining preserves random walk dynamics. 
Applying this framework to three major renormalization methods across real-world networks, we discover that each approach produces a distinct \emph{dynamical fingerprint}---a characteristic harmonic degree curve revealing the underlying physics of how it aggregates structure. Remarkably, we identify networks where Laplacian renormalization spontaneously produces exact harmonic morphisms at specific scales, achieving exact first-exit random-walk preservation that entropic susceptibility~\cite{Villegas_Laplacian_renorm_heterogeneous_2023,nurisso2025higher} fails to detect.

The paper proceeds as follows. 
In Sec.~\ref{section:harmonic_morphisms}, we develop the theory of harmonic morphisms and prove our main theorem linking them to random walk preservation. 
In Sec.~\ref{section:harmonic_degree}, we define the harmonic degree metrics. 
In Sec.~\ref{section:RG}, we apply these tools to three major renormalization methods, revealing distinct dynamical fingerprints. 
In Sec.~\ref{section:laplacian_deep}, we examine Laplacian renormalization in depth, including the discovery of exact harmonic morphisms and the comparison with entropic susceptibility. 
In Sec.~\ref{section:higher-order}, we explore extensions to higher-order networks.
Sec.~\ref{section:discussion} discusses implications and open questions.

%% ====================================================================
%%  THEORETICAL FRAMEWORK
%% ====================================================================
\section{Theoretical Framework}
\label{section:harmonic_morphisms}

\textbf{\begin{figure*}[t]
    \centering
    \includegraphics[width=\linewidth]{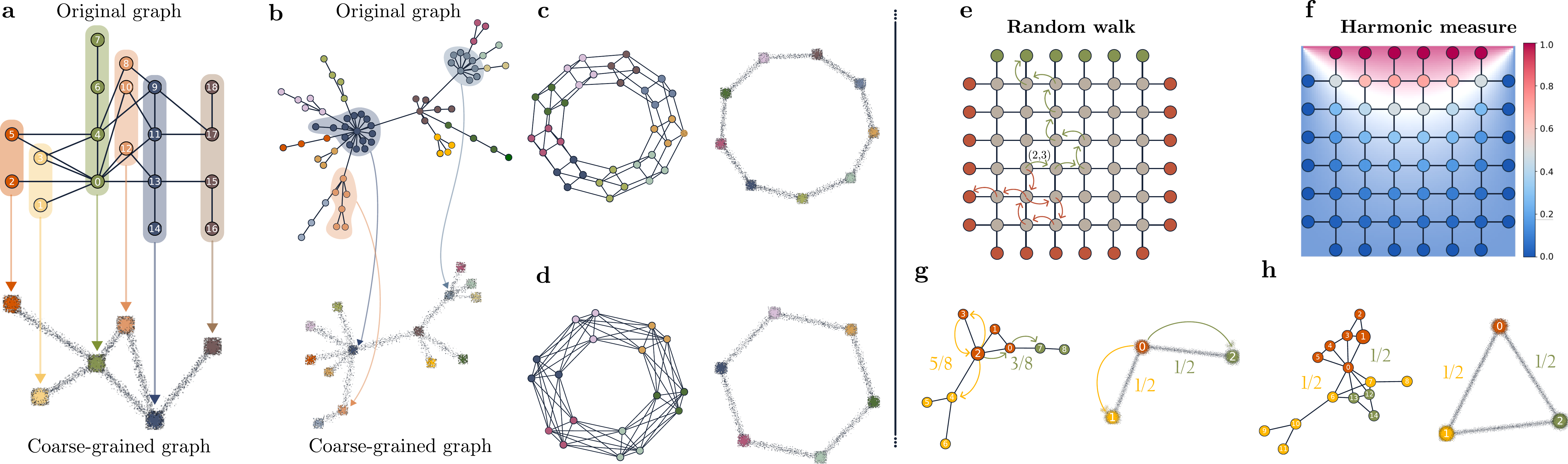}
    \caption{\textbf{a.}~The original example of a horizontally conformal function from %Ref.
    ~\cite{Urakawa_Harmonic_morphisms_2000}. Node~$0$ has the same number of neighbors in each adjacent macro-set: 2 yellows, 2 reds, 2 light browns, and 2 blues. Node~13 has exactly one neighbor in each adjacent macro-set. \textbf{b.}~Harmonic morphism on a tree-like graph via node collapse: peripheral nodes merge into their hub. \textbf{c.}~Harmonic morphism from a torus-like graph onto a cycle. \textbf{d.}~A combinatorial conformal map (Definition~\ref{def:combinatorial_conformal}) from a torus to a cycle: each pink node has 3 blue, 3 yellow, and 3 pink neighbors (counting itself), yielding constant multiplicities including the self-loop. \textbf{e.}~Grid graph $G$ with boundary $\partial G$ (red and green nodes); $B \subseteq \partial G$ consists of the green nodes. Example random walk starts from $(2,3)$.
    \textbf{f.}~Harmonic measure $\mathcal{H}(\cdot,B)$ solving the discrete Laplace equation; the continuous limit is shown in the background.
    \textbf{g.}~A coarse-graining that is \emph{not} a harmonic morphism: node~2 has one yellow neighbor but no green ones, so a walker starting from node~2 exits to the yellow macro-set with probability $5/8$ and to the green one with probability $3/8$, not the $1/2$ expected from the coarse-grained graph.
    \textbf{h.}~A coarse-graining that \emph{is} a harmonic morphism: to exit the red macro-set from node~1, the walker must pass through node~0, which has equal numbers of green and yellow neighbors, guaranteeing balanced exit probabilities.}
    \label{fig:figure_1}
\end{figure*}}

\subsection{Graph coarse-graining as a surjective function}

The starting point of our approach is to treat any coarse-graining as a surjective map $\varphi: V \rightarrow \mathcal{V}$ that compresses a graph $G=(V,E)$ into a coarser graph $\mathcal{G}=(\mathcal{V},\mathcal{E})$. We refer to each $y \in \mathcal{V}$ as a \emph{macro-node} and to its pre-image $\varphi^{-1}(y)$ as a \emph{macro-set}. Over each graph we consider the respective simple random walk. Our goal is to identify which properties $\varphi$ must satisfy so that the random walk on $G$, when observed only at the times it crosses from one macro-set to another, is faithfully represented by the random walk on $\mathcal{G}$. 
Figure~\ref{fig:figure_1} illustrates several examples of such maps (which will become clearer in the next section). 

Throughout, we focus on finite, undirected, unweighted graphs, though our results extend to locally finite graphs (i.e. finite degrees even when $V$ is infinite). The random walk Laplacian operator $\Delta$ acts on $f: V \rightarrow \mathbb{R}$ as
\begin{equation}
\Delta f(v) = f(v) - \frac{1}{\deg(v)} \sum_{w \sim v} f(w),
\end{equation}
where $w \sim v$ denotes adjacency. A function $f$ is \emph{harmonic at node $x$ }if $\Delta f(x) = 0$, meaning $f(x)$ equals the average of $f$ over the neighbors of $x$:
\begin{equation}
f(x) = \frac{1}{\deg(x)}\sum_{z \sim x}f(z).
\end{equation}
Harmonic functions are intimately tied to random walks: a harmonic function on a domain with prescribed boundary values gives the expected value of the boundary data under the first-exit distribution of the walk~\cite{Lawler1991,barahudcovaLatticeModels}. 
This connection is what makes the following notion dynamically meaningful.

\subsection{Harmonic morphisms and horizontal conformality}
Relying on Urakawa's work on discrete analogs of harmonic morphisms between Riemannian manifolds~\cite{Urakawa_Harmonic_morphisms_2000,Fuglede1978HarmonicMB,Ishihara1979AMO}, we say that a surjective map $\varphi: V \rightarrow \mathcal{V}$ is a \emph{harmonic morphism} if it preserves harmonicity locally: whenever a function on $\mathcal{G}$ is harmonic at a macro-node $y$, its pullback through $\varphi$ is harmonic at every node in the corresponding macro-set $\varphi^{-1}(y)$.

\begin{definition}[Harmonic morphism]
\label{def:harmonic_morphism}
Given two graphs $G=(V,E)$ and $\mathcal{G}=(\mathcal{V},\mathcal{E})$, a function $\varphi: V \rightarrow \mathcal{V}$ is a \emph{harmonic morphism} if: $\forall y \in \varphi(V)$, $\forall x \in V$ such that $\varphi(x)=y$, and for any function $f:\mathcal{V} \rightarrow \mathbb{R}$ that is harmonic at $y$, the composition $f \circ \varphi : V \rightarrow \mathbb{R}$ is harmonic at $x$.
\end{definition} 

This abstract condition admits a concrete, checkable characterization. Urakawa's main result~\cite{Urakawa_Harmonic_morphisms_2000} shows that surjective harmonic morphisms are equivalent to \emph{horizontally conformal} functions, maps satisfying two conditions: (i) adjacent nodes must map to adjacent or identical macro-nodes (a \emph{weak homomorphism}), and (ii) for every node $x$, the number of its neighbors falling in each adjacent macro-set must be the same across all such macro-sets.

\begin{definition}[Horizontally conformal function]
\label{def:horizontal_conformal}
Given graphs $G=(V,E)$ and $\mathcal{G}=(\mathcal{V},\mathcal{E})$ and a function $\varphi: V \rightarrow \mathcal{V}$, we say $\varphi$ is \emph{horizontally conformal} if:
\begin{enumerate}
\item[(i)] $\forall x,z \in V$: $x \sim z \implies \varphi(x) \sim \varphi(z) \text{ or } \varphi(x)=\varphi(z)$

\item[(ii)] $\forall y \in \mathcal{V}$, $\forall x \in \varphi^{-1}(y)$: the quantity $|\{ z\in \varphi^{-1}(y') : z \sim x \}|$ is constant for all $y' \sim y$.
\end{enumerate}
\end{definition}

\begin{theorem}[Urakawa~\cite{Urakawa_Harmonic_morphisms_2000}]
\label{thm:urakawa}
Let $G=(V,E)$ and $\mathcal{G}=(\mathcal{V},\mathcal{E})$ be two graphs and $\varphi: V \rightarrow \mathcal{V}$ a surjective function. Then $\varphi$ is a harmonic morphism if and only if $\varphi$ is horizontally conformal.
\end{theorem}

Condition~(ii) is the key symmetry. It requires that macro-sets be connected to one another in a balanced way, a \emph{mixing symmetry}. To see what this means concretely, consider the Urakawa example in Fig.~\ref{fig:figure_1}(a): node~0 has exactly 2 neighbors in each of the four adjacent macro-sets (yellow, red, light brown, blue), while node~13 has exactly 1 neighbor in each. These constant multiplicities ensure horizontal conformality and thus, by Theorem~\ref{thm:urakawa}, that $\varphi$ is a harmonic morphism. Other natural examples include graph collapses [Fig.~\ref{fig:figure_1}(b)], where peripheral nodes connected only through a hub are merged into it---a structure that arises naturally in tree-like networks---and torus-to-cycle projections [Fig.~\ref{fig:figure_1}(c)].

\subsection{Harmonic morphisms preserve random walk dynamics}

We now turn to the dynamical content of harmonic morphisms. 
Urakawa himself claimed that discrete harmonic morphisms project random walks on a graph to another by a suitable time change (~\cite{Urakawa_Harmonic_morphisms_2000}, Remark $2.3$), without formalizing or demonstrating this statement. 
Here we rigorously state this proposition and provide a formal proof of it.
Our key tool is the \emph{harmonic measure} from discrete potential theory~\cite{barahudcovaLatticeModels,Lawler1991}: given a graph $G$ with boundary $\partial G$, the harmonic measure $\mathcal{H}(x, B) = \mathbb{P}(X^x_\tau \in B)$ is the probability that a random walk starting at $x$ exits $G$ through the boundary subset $B \subseteq \partial G$ at its first exit time $\tau$. Crucially, $\mathcal{H}(\cdot, B)$ is the unique harmonic function on $G$ with boundary values 1 on $B$ and 0 on $\partial G \setminus B$~\cite{barahudcovaLatticeModels}. Figures~\ref{fig:figure_1}(e,f) illustrate this on a grid graph: the harmonic measure, visualized as a heatmap, solves the discrete Laplace equation and converges to its continuous counterpart in the scaling limit.

\begin{definition}[Harmonic measure]
\label{def:harmonic_measure}
Let $G=(V,E)$ be a graph with boundary $\partial G$, where $G$ is seen as a subgraph of $G \sqcup \partial G$. Let $\{X^x_n\}$ be a simple random walk starting at $x \in G$, and let $B \subseteq \partial G$. The \emph{harmonic measure} of $x$ with respect to $B$ is
$\mathcal{H}(x,B) = \mathbb{P}(X^x_\tau \in B)$,
where $\tau = \inf\{ n > 0 : X^x_n \in \partial G \}$ is the first exit time from $G$.
\end{definition}
To connect this to coarse-graining, consider a surjective $\varphi: V \rightarrow \mathcal{V}$ inducing a partition of $V$ into macro-sets. For each macro-node $y \in \mathcal{V}$, let $G_y$ be the induced subgraph on $\varphi^{-1}(y)$, let $\partial G_y$ be the set of nodes outside $\varphi^{-1}(y)$ adjacent to it, and let $\partial G_{yy'} = \partial G_y \cap \varphi^{-1}(y')$ partition this boundary by macro-set, where $y'\in\mathcal{V}\backslash y$. We write $\mathcal{H}_y(x, y') := \mathcal{H}_{G_y}(x, \partial G_{yy'})$ for the probability that a walker starting at $x \in \varphi^{-1}(y)$ first exits macro-set $y$ by entering macro-set $y'$.

With this notation, we can state our main result, which gives harmonic morphisms a precise dynamical meaning: they are exactly the coarse-grainings under which first-passage probabilities across macro-sets match one-step transition probabilities in the coarse-grained graph.

\begin{theorem}[Random walk preservation]
\label{theo:RW_preservation}
Let $G=(V,E)$ and $\mathcal{G}=(\mathcal{V},\mathcal{E})$ be two graphs, and let $\varphi: V \rightarrow \mathcal{V}$ be a surjective function. 
Let $Y$ denote the simple random walk on $\mathcal{G}$. 
Then $\varphi$ is a harmonic morphism if and only if
\begin{eqnarray}
\forall y \in \mathcal{V}, \, \forall x \in \varphi^{-1}(y), \, \forall y' \sim y: \quad \nonumber \\ 
\mathcal{H}_y(x, y') = \mathbb{P}(Y_{t=1} = y' \mid Y_{t=0} = y) = \frac{1}{\deg(y)}.
\end{eqnarray}
\end{theorem}

\begin{proof}[Proof sketch]
For the full proof, see the Supplementary Material \ref{sec:proofs_sm}.

($\Rightarrow$) Assume $\varphi$ is a harmonic morphism. 
Conditioning on the last step before exit, horizontal conformality (condition~ii of Definition~\ref{def:horizontal_conformal}) ensures the probability of exiting to $\partial G_{yy'}$ equals $1/\deg(y)$, independent of the exit node.

($\Leftarrow$) Suppose the harmonic measure condition holds. First, $\varphi$ must be a weak homomorphism (otherwise harmonic measures would sum to more than~1). Then, using that $\mathcal{H}_y(\cdot, y')$ is harmonic on $G_y$, the hypothesis $\mathcal{H}_y(x,y') = 1/\deg(y)$ for all $y' \sim y$ implies all multiplicities $k_{y'}$ are equal, yielding horizontal conformality.
\end{proof}

The difference between a harmonic and a non-harmonic coarse-graining is illustrated in Fig.~\ref{fig:figure_1}(g,h). 
In panel~(g), node~2 has one yellow neighbor but no green ones; a walker starting from node~2 therefore exits to the yellow macro-set with probability $5/8$, not the $1/2$ predicted by the coarse-grained graph. 
In panel~(h), the coarse-graining is a harmonic morphism: any walker exiting the red macro-set must pass through node~0, which has equal numbers of green and yellow neighbors, guaranteeing exit probabilities of exactly $1/2$ to each.

Theorem~\ref{theo:RW_preservation} has a natural physical interpretation. 
All dynamics \emph{within} macro-sets are disregarded: we track only when the walker crosses from one macro-set to another, and which one it enters. 
These crossings occur after a random time $\tau(x, \varphi)$ that depends on both the starting node and its macro-set, so $\varphi$ simultaneously performs a spatial scale change (compressing the network) and a temporal scale change (through $\tau$). 
This temporal flexibility is what distinguishes our framework from approaches to coarse-graining Markov chains that track internal macro-state activity at every time step~\cite{Faccin_2017,Consensus_coarse_graining_Barahona,rosas2024softwarenaturalworldcomputational}. 
The node collapse of Fig.~\ref{fig:figure_1}(b) is an instructive example: peripheral nodes may take many steps to reach the hub, while already-central nodes transition in one step, yet the harmonic morphism treats both on equal footing by focusing solely on first-exit probabilities. 
%More broadly, a walker might traverse a dense cluster in expected time $O(\log n)$ while crossing between well-separated macro-sets requires $O(n)$ steps, yet harmonic morphisms accommodate both time scales seamlessly.

\subsection{Combinatorial conformality and lazy random walks}
Horizontal conformality constrains how macro-sets connect to \emph{other} macro-sets (the ``horizontal'' structure). 
Inspired by conformal symmetry in physics~\cite{Introduction_to_Conformal_Field_Theory,Conformal_Invariance_Lattice_Models_Smirnov}, we can impose a stricter condition that also constrains connections \emph{within} macro-sets (the ``vertical'' structure). 
The resulting notion, which we call \emph{combinatorial conformality}, additionally requires the number of neighbors a node has in its own macro-set (counting itself) to participate in the same constant-multiplicity condition.

\begin{definition}[Combinatorial conformal function]
\label{def:combinatorial_conformal}
A function $\varphi: V \rightarrow \mathcal{V}$ is \emph{combinatorial conformal} if it satisfies condition~(i) of Definition~\ref{def:horizontal_conformal} and, additionally, $|\{ z\in \varphi^{-1}(y') : z \sim x \text{ or } z = x \}|$ is constant for all $y'$ with $y' \sim y$ or $y' = y$.
\end{definition}

Figure~\ref{fig:figure_1}(d) shows a concrete example: a torus mapped to a cycle where every pink node has 3~blue neighbors, 3~yellow neighbors, and 2~pink neighbors plus itself (3~total counting the self-loop)---the same multiplicity in every direction, including within its own macro-set. 
This additional constraint relates to lazy random walks, where at each step the walker either moves to a neighbor or stays put with uniform probability $1/(\deg(v)+1)$.

\begin{theorem}[Lazy random walk preservation]
\label{thm:lazy_rw}
A function $\varphi$ is combinatorial conformal if and only if it preserves lazy random walk transition probabilities: 
$\mathbb{P}(\mathcal{L}_{t=1} \in \varphi^{-1}(y') \mid \mathcal{L}_{t=0} = x) = 1/(\deg(y)+1)$ for all $y'$ with $y' \sim y$ or $y'=y$.
\end{theorem}
For the full proof see Supplementary Material Section \ref{sec:proofs_sm}. 
Unlike harmonic morphisms, combinatorial conformal transformations involve \emph{no} temporal scale change: the discrete time steps of the lazy walk align between original and coarse-grained networks, making them strongly lumpable in the sense of Ref.~\cite{rosas2024softwarenaturalworldcomputational}. 
This is a much more restrictive condition---non-trivial collapses are never combinatorial conformal, and tree-like structures admit no non-trivial examples with connected pre-images---but when it holds, it provides a particularly clean form of dynamical equivalence. 
We expand on the connections between our framework and conformal field theory in the Supplementary Material~\ref{si:link-to-conformal-field-theory}.

\begin{figure*}
    \centering
    \includegraphics[width=\linewidth]{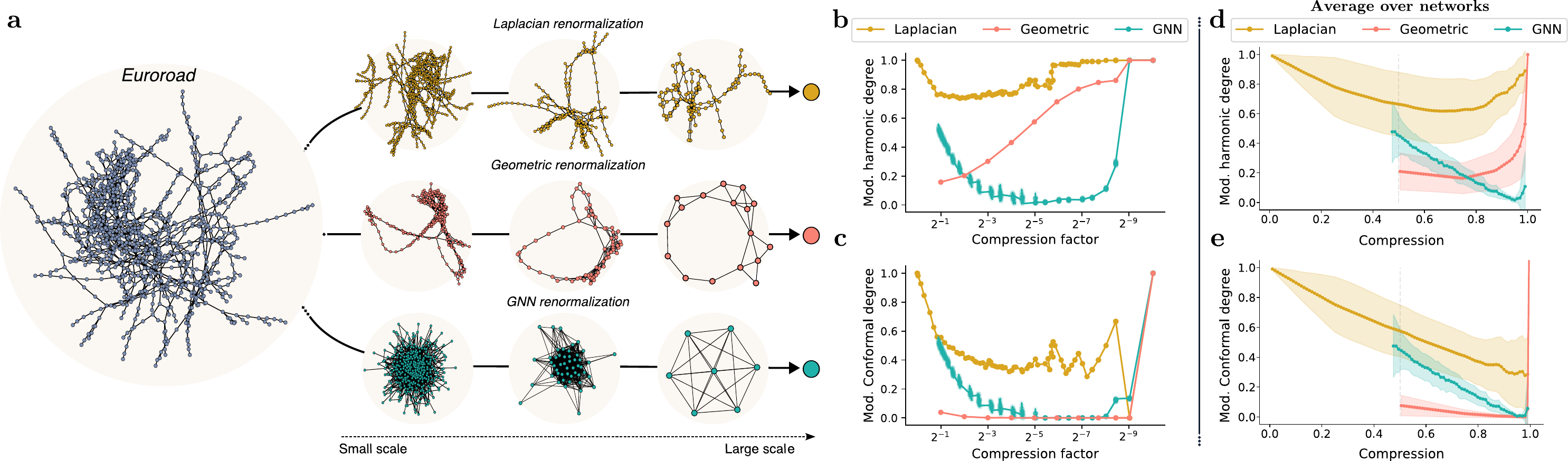}
    \caption{\textbf{Example of renormalizations of the Euro-Road network.} \textbf{a.}~Instances from the RG flow under geometric, Laplacian, and GNN-based renormalization methods. 
    \textbf{b,c.}~Modified harmonic degree (b) and modified conformal degree (c) as a function of network compression for the Euro-Road network.
    \textbf{d,e.}~Modified harmonic degree (d) and conformal degree (e) as a function of compression, averaged over 16 networks, for the three renormalization methods. 
    The distinct fingerprints observed on the Euro-Road network generalize robustly: geometric renormalization (S-curve), Laplacian (high-low-high), GNN (uniformly low, with a drop from moderate values to low ones).
    The compression is shown in logarithmic scale.
    }
    \label{fig:euroroad_comparison}
\end{figure*}

%% ====================================================================
%%  HARMONIC DEGREE
%% ====================================================================
\section{Harmonic degree metrics}
\label{section:harmonic_degree}

Given an arbitrary coarse-graining $\varphi: V \rightarrow \mathcal{V}$, we quantify how close it is to a harmonic morphism via the computationally tractable characterization of horizontal conformality. 
For each node $x$ with $\varphi(x)=y$, we compute
\begin{equation}
k_{y'}(x) := |\{ z \in \varphi^{-1}(y') : z \sim x \}|
\label{eq:multiplicity}
\end{equation}
for each $y' \sim y$. If all $k_{y'}(x)$ are equal, $x$ is a \emph{harmonic node}. Given a set $A \subseteq V$, we call $H_{\varphi}(A)$ its subset of harmonic nodes with respect to $\varphi$, we simply indicate it as $H(A)$ when $\varphi$ is clear from the context. 
We introduce three complementary metrics.
\begin{enumerate}
\item  \textbf{Mean harmonic degree:} defined as 
\begin{equation}
H_{\text{mean}} = \frac{|H(V)|}{|V|}
\end{equation}
corresponding to the fraction of harmonic nodes.

\item \textbf{Modified harmonic degree:} To prevent large macro-sets from dominating, we average harmonicity within each macro-set:
\begin{equation}
H_{\text{mod}} = \frac{1}{|\mathcal{V}|}\sum_{y \in \mathcal{V}} \frac{|H(\varphi^{-1}(y))|}{|\varphi^{-1}(y)|}.
\end{equation}
This is our primary metric for renormalization analysis.

\item \textbf{Harmonic deviation:} which provides a continuous measure of multiplicity imbalance across macro-sets:
\begin{equation}
H_{\text{Dev}} = \frac{1}{|V|}\sum_{x \in V} \mathrm{std}\!\left(\{k_{y'}(x) : y' \sim \varphi(x)\}\right).
\end{equation}
\end{enumerate}

We define analogous conformal degree metrics ($CF_{\text{mean}}, CF_{\text{mod}}, CF_{\text{Dev}}$) based on Definition~\ref{def:combinatorial_conformal}; explicit definitions appear in the Supplementary Material~\ref{sec:conformal_metrics_sm} and ~\ref{sec:conformal_metrics_sm}). 
% In practice, combinatorial conformality is extremely restrictive: $CF_{\text{mod}} = 0$ in approximately 85\% of network--method combinations (see Supplementary Material~\ref{sec:conformal_metrics_sm}). 
% The rare exceptions, primarily singleton clusters in Infomap partitions, indicate nodes isolated from diffusion-driven aggregation.

In the following, we define the \textit{compression} $\eta$ of a coarse-graining as
\begin{equation}
\eta(\varphi) = 1 - \frac{|\mathcal{V}|}{|V|},
\end{equation}
that is, $1$ minus the ratio between the number of nodes of the coarse-grained ($|\mathcal{V}|$) and original networks ($|V|$).

As a preliminary validation before turning to our main results on renormalization flows, in Supplementary Material~\ref{sec:clustering_sm} we benchmark these metrics against five community detection algorithms on numerous real networks. 
Harmonic degree varies across both networks and methods, with no algorithm universally maximizing it, confirming that it captures complementary meso-scale information not reducible to standard clustering quality.

% Before moving to the main contribution of this work, we explore these metrics by benchmarking them against five community detection algorithms on approximately 50 real networks from diverse domains (social, infrastructure, biological, tech/web, animal, and collaboration). 
% The results, presented in the Supplementary Material~\ref{sec:clustering_sm}, reveal that harmonic degree is both network and clustering method dependent (no algorithm universally maximizes it)
% %, is distinct from modularity (high modularity does not guarantee high harmonicity),
% and captures complementary meso-scale information about network organization.

%In the following, we define the \textit{compression} $\eta$ of a coarse-graining as
%\begin{equation}
%\eta(\varphi) = 1 - \frac{|\mathcal{V}|}{|V|},
%\end{equation}
%that is, $1$ minus the ratio between the number of nodes of the coarse-grained ($|\mathcal{V}|$) and original networks ($|V|$).

%% ====================================================================
%%  RESULTS: RENORMALIZATION FINGERPRINTS
%% ====================================================================
\section{Dynamical fingerprints of renormalization}
\label{section:RG}

Having established the theoretical framework and its diagnostic metrics, we now move to the central contribution of this work, by assessing to what degree existing network renormalization schemes align with random walk dynamics across scales. We focus on three major approaches, each built on different physical assumptions about what coarse-graining should preserve.

\emph{Geometric renormalization}~\cite{Multiscale_unfolding_of_real_nets_geometric_renorm_2018} embeds a network into hyperbolic space (via the Mercator algorithm~\cite{Mercator, D_Mercator}) then iteratively pairs nodes by angular proximity. 
At each iteration, roughly half the nodes are merged, exploiting the observation that many real networks exhibit latent hyperbolic geometry~\cite{Hyperbolic_geometry_complex_networks,Self_similarity_brain,Barjuan2025Multiscale}. 
The method is inherently top-down: a single global embedding guides all merging decisions at every scale.

\emph{Laplacian renormalization}~\cite{Villegas_Laplacian_renorm_heterogeneous_2023,nurisso2025higher} takes the opposite approach: merging decisions are driven by local diffusion. 
Given the graph Laplacian $L$, one computes the density matrix $\rho(t) = e^{-tL}/\mathrm{Tr}(e^{-tL})$ and merges nodes $i$ and $j$ when the information they exchanged at time $t$ exceeds a threshold: $\rho_{ij}(t) \geq \min(\rho_{ii}(t), \rho_{jj}(t))$, where the diffusion time $t$ acts as a continuous scale parameter. 
The method is inherently bottom-up: local diffusion properties alone determine which nodes are grouped. 
In practice, for numerical stability we consider the merging threshold as $\rho_{ij}(t) \geq \min(\rho_{ii}(t), \rho_{jj}(t))-\varepsilon$, where $\varepsilon$ is a control parameter usually chosen to be $10^{-3}, 10^{-4} \text{ or }0$ (see  Supplementary Material~\ref{sec:computational_sm} for implementation details, and Supplementary Material~\ref{sec:eq_laplacian_sm} for an extension to a recent variant, the Equilibrium Laplacian Renormalization~\cite{yi2025equilibriumpreservinglaplacianrenormalizationgroup}).

\emph{GNN-based renormalization}~\cite{Coarse_graining_network_De_Domenico} trains a graph neural network to produce coarse-grainings that preserve the heat-trace partition function $Z(t) = \mathrm{Tr}(e^{-tL})$, optimizing an objective based on $Z(t)/N \approx Z'(t)/N'$ for all $t \geq 0$. 
This optimizes for eigenvalue spectrum preservation rather than local diffusion or geometric proximity. Since the GNN produces soft assignments, for each network we sample 25 hard partitions and report statistics. Following the original paper and repository, we considered $2$ layers, with hidden dimension$= 64$.

\begin{figure*}
    \centering
    \includegraphics[width=\linewidth]{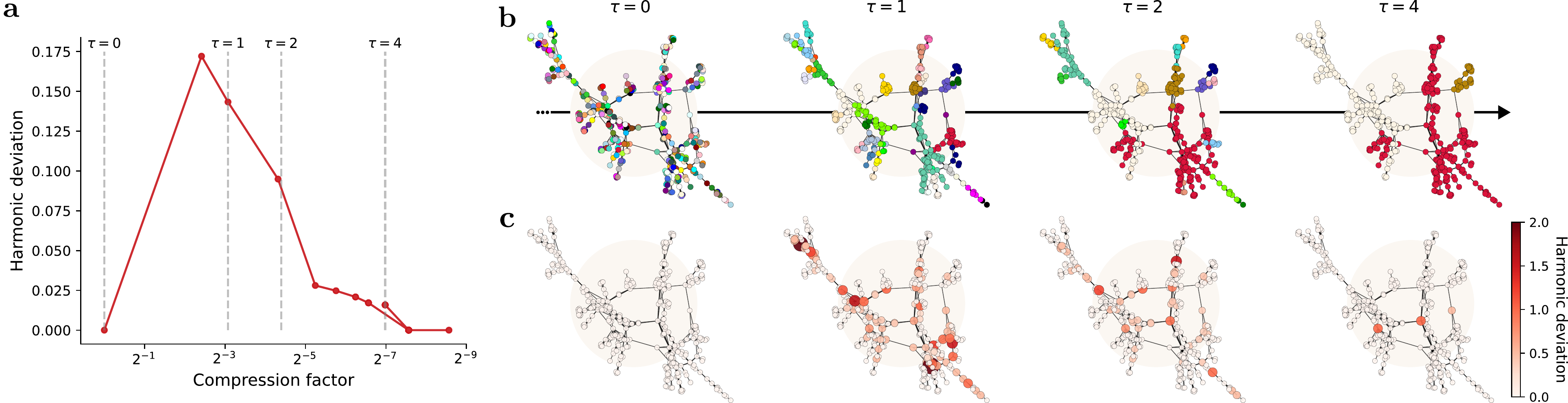}
    \caption{\textbf{Spatial distribution of harmonic deviation under Laplacian renormalization of the NetSci collaboration network.}
    \textbf{a.}~Harmonic deviation curve through Laplacian renormalization.
    \textbf{b.}~Community assignments at four diffusion times $t$; colors denote macro-set membership, shaded regions highlight emerging clusters.
    \textbf{c.}~Per-node harmonic deviation $\sigma_x$ on the same layouts. White nodes ($\sigma_x \approx 0$) satisfy horizontal conformality; dark red nodes ($\sigma_x \geq 2$) have strongly imbalanced boundary degrees.
    At small $t$, nearly all nodes are harmonic. At intermediate $t$, high-deviation nodes concentrate at the boundaries between merging communities. At large $t$, deviation returns to low values as the network coalesces into few, well-separated macro-sets.
    }
    \label{fig:hdev_viz}
\end{figure*}

In Figure~\ref{fig:euroroad_comparison} we report the results of applying these three methods to the Euro-Road network. 
We track the modified harmonic and conformal degrees across scales, and find that the dynamical fingerprints are strikingly distinct. 
Geometric renormalization produces a characteristic \textit{S-shaped curve}: $H_{\text{mod}} \approx 0.20$ at early iterations (the first and the second, $\eta = 0.5 \text{ and } \eta = 0.75$), rising above $H_{\text{mod}} = 0.8$ at late iterations (from iteration $7$, $\eta= 0.99$). 
The globally imposed embedding conflicts with local road connectivity at small scales, as angularly close nodes may not mix symmetrically with their local neighborhoods, but successfully captures the coarse geographic organization at large scales, where regional clusters interconnect more symmetrically. 
Conformality remains close to null ($CF_{\text{mod}} \approx 0$) throughout.

Laplacian renormalization produces instead a markedly different \textit{high-low-high pattern}: 
$H_{\text{mod}} = 0.83$ at compression $\eta= 0.33$ ($t=1.2$), a shallow minimum at $H_{\text{mod}} = 0.73$ at compression $\eta=0.84$ ($t \simeq 4$), and then recovering high values of harmonicity ($H_{\text{mod}} = 0.77$) for high compression $\eta\geq0.9$ ($t = 5.5$). 
At small compression $\eta$, diffusion is localized and nodes merge only with strongly-connected neighbors, forming micro-clusters that interface symmetrically with the rest of the network. 
At intermediate $\eta$, diffusion spreads unevenly, some regions consolidate faster than others, creating asymmetric merging. 
At large $\eta$, the network coalesces into a few well-separated macro-sets and harmonicity recovers naturally.
The conformal degree follows a similar pattern despite having overall lower values ($\sim 0.4$--$0.5$).  

Finally, GNN-based renormalization produces uniformly low harmonicity: $H_{\text{mod}}$ drops from $0.5$ to almost $0$ as compression increases. 
The method groups nodes by structural role regardless of whether they form assortative macro-sets, producing asymmetric multiplicities. 
Notably, $H_{\text{mod}} \approx CF_{\text{mod}}$ throughout, indicating that what little harmonicity exists is driven by singleton identity mappings rather than genuine balanced mixing ---in sharp contrast with Laplacian renormalization, where $H_{\text{mod}} \gg CF_{\text{mod}}$.

Crucially, these fingerprints are not specific to the Euro-Road network, but rather broadly shared across networks. 
Figure~\ref{fig:euroroad_comparison}(d) shows $H_{\text{mod}}$ averaged over 16 networks: the S-shaped geometric curve, the high-low-high Laplacian pattern, and the uniformly low GNN signature are all clearly reproduced across networks, showing that --despite variability at the level of individual networks-- each renormalization scheme produces a different effect on the coarse-grained dynamics, as expected by their different underlying physical assumptions.
On average, we find that Laplacian renormalization is the method that best respects the diffusion structure of networks across all scales, in terms of both harmonic and conformal metrics. 
However, individual networks can display some heterogeneity in the ranking of the various curves. In the Supplementary Material~\ref{sec:detailed_renorm_sm} we provide the disaggregated analyses for each network.

To understand the structural origin of the Laplacian high-low-high fingerprint, we examine the \emph{spatial distribution} of harmonic deviation on an example network, the NetSci collaboration network (Figure~\ref{fig:hdev_viz}).
Panel~(b) shows the community assignments at four diffusion times $t$, while panel~(c) maps the per-node harmonic deviation $\sigma_x$ onto the same layout (white: $\sigma_x = 0$, dark red: $\sigma_x = 2$).
At early scales ($t = 0$), communities are small and tightly connected: nearly all nodes are internal to their macro-set and thus trivially harmonic, yielding uniformly low deviation.
As diffusion progresses ($t = 1$--$2$), communities grow and merge unevenly: nodes at the boundaries between newly formed macro-sets ---particularly at branching points of the network--- develop high deviation, reflecting asymmetric multiplicities toward different neighboring communities.
These high-deviation boundary nodes are the structural origin of the dip in $H_{\text{mod}}$ at intermediate compression.
At large scales ($t = 4$), the network has coalesced into few well-separated clusters; most nodes are again deep inside their community, and the sparse inter-cluster interfaces involve few boundary nodes with balanced connectivity, so deviation drops back to low values.
This confirms that the temporary loss of harmonicity at intermediate scales is a \emph{localized} phenomenon ---a transient imbalance at community boundaries during the merging regime--- rather than a global degradation of the partition quality.

The analysis above reveals that Laplacian renormalization consistently achieves the highest harmonic degree among the three methods, with values remaining above $0.7$ across most of the compression range. A natural question follows: can Laplacian renormalization achieve \emph{perfect} harmonicity---exact harmonic morphisms---in real networks? And if so, what structural properties enable this?

%% ====================================================================
%%  RESULTS: LAPLACIAN DEEP DIVE
%% ====================================================================
\section{Structural scales and dynamical preservation under Laplacian renormalization}
\label{section:laplacian_deep}

Among the three methods, Laplacian renormalization consistently achieves the highest harmonic degree across scales and network types. This motivates a closer examination of two remarkable properties: the spontaneous emergence of exact harmonic morphisms, and the relationship between harmonic degree and entropic susceptibility.

\subsection{Exact harmonic morphisms}

Our most striking empirical finding is that Laplacian renormalization can, at specific scales, produce coarse-grainings that satisfy horizontal conformality exactly, and therefore define exact harmonic morphisms. 
We verify this combinatorially, node by node: every node satisfies condition~(ii) of Definition~\ref{def:horizontal_conformal} across the renormalization process (across $ \varepsilon = 10^{-3}, 10^{-4} \text{ or } 0$), yielding $H_{\text{mod}} = 1$ (plotted values are rounded to three decimal places). 
We observe this behavior, with harmonic degree curves nearly constant at unity, in Facebook, Web-edu, CS Collab, and Yeast (see Supplementary Material~\ref{sec:detailed_renorm_sm} for full visualization). 

The Facebook network \cite{traud2012social} of \Cref{fig:figure_4}(a) provides a particularly instructive example. 
At diffusion time $t=1.0$, with $\varepsilon = 0$, the coarse-graining is an exact harmonic morphism with a characteristic structure: large groups of nodes connected to hubs collapse into macro-nodes, while single bridge nodes map to themselves in the coarse-grained graph. 
The resulting macro-sets exhibit balanced boundary structure---constant connection multiplicities across neighboring clusters---while the self-mapped bridge nodes maintain nearly uniform ties to their surrounding macro-sets. This is a multi-scale transformation in the precise sense of Sec.~\ref{section:harmonic_morphisms}, since groups of nodes at different hierarchical levels are placed on equal footing in the coarse-graining.

We hypothesize that exact harmonic morphisms emerge when diffusion-based merging produces cohesive macro-sets together with bridge nodes having balanced external multiplicities. Specifically, the merging criterion naturally respects the symmetries required for horizontal conformality: tightly connected groups merge early into cohesive macro-sets; nodes with balanced external connectivity resist merging (their diffusion spreads uniformly, so no single neighbor exceeds the threshold), becoming bridge nodes; at the right time scale, these two populations coexist with balanced boundary multiplicities.

\begin{figure}
    \centering
    \includegraphics[width=\linewidth]{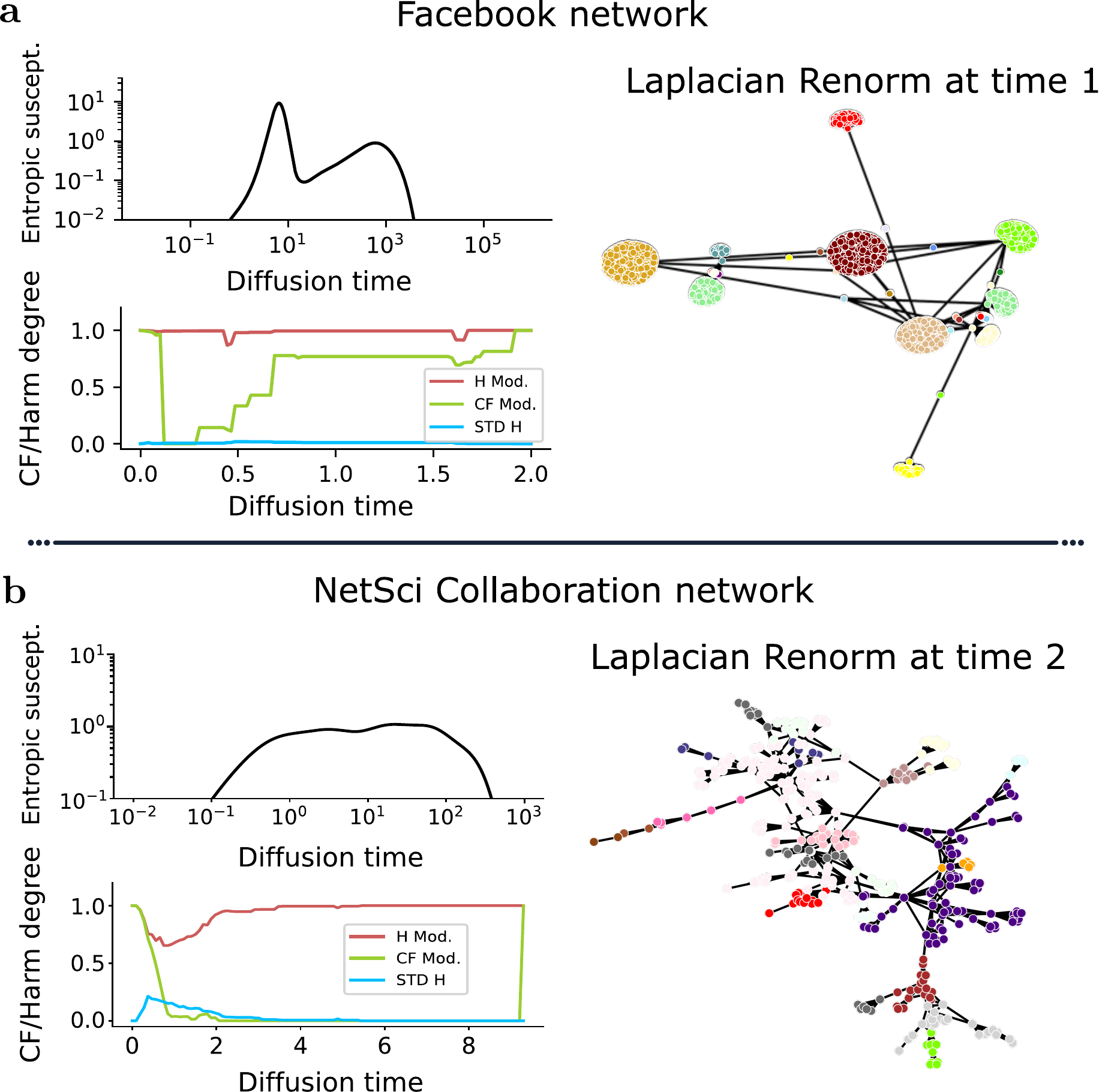}
    \caption{\textbf{Harmonic degree and entropic susceptibility.} 
    \textbf{a.}~The Facebook network \cite{traud2012social} maintains $H_{\text{mod}} \approx 1$ 
    throughout the Laplacian renormalization flow ($t \in [0,2]$), 
    indicating exact random-walk preservation, while the entropic 
    susceptibility $C(t)$ reveals two distinct structural scales 
    ($t \approx 10$ and $t \approx 10^3$). 
    \textbf{b.}~The NetSci collaboration network shows scale-invariant 
    diffusion (flat $C(t)$) but lower harmonic degree. 
    The two diagnostics are complementary: entropic susceptibility 
    characterizes the network's diffusion landscape, while harmonic 
    degree evaluates whether the coarse-graining induced at each 
    scale preserves random-walk dynamics.}
    \label{fig:figure_4}
\end{figure}

\subsection{Harmonic degree versus entropic susceptibility}

A natural question is how harmonic degree relates to existing diagnostics for Laplacian renormalization. 
The entropic susceptibility~\cite{Villegas_Laplacian_renorm_heterogeneous_2023,nurisso2025higher} $C(t) = -\dd S(t)/\dd \log t$ detects characteristic scales by measuring diffusion deceleration: local maxima indicate meso-scale structure, plateaus signal scale invariance. 
Comparing it with harmonic degree reveals a crucial distinction: \emph{entropic susceptibility detects scales that diffusion perceives; harmonic degree detects whether the transformation induced by diffusion preserves random walks}. 
These are related but fundamentally different properties.

For the Facebook network \cite{traud2012social} of \Cref{fig:figure_4}(a), $C(t)$ reveals two structural scales ($t \approx 10$ and $t \approx 10^3$), yet $H_{\text{mod}} \approx 1.0$ throughout the renormalization process in $t \in [0, 2]$.
The bridge nodes, which $C(t)$ does not explicitly track, balance the transformation despite apparent scale separation. 
More broadly, across our dataset we observe all possible combinations: scale-invariant diffusion with low harmonicity (Euroroad, NetSci Collab, Bcs09 Power Network) (see \Cref{fig:figure_4}(b) and Supplementary Material~\ref{sec:detailed_renorm_sm}), multi-scale structure with high harmonicity (Facebook, Web: education), high harmonicity with scale invariance (CS Collab and Bio: Plant to some extent) and several cases where we have neither high harmonicity nor scale invariance.
Harmonic degree and entropic susceptibility thus provide complementary information: the former characterizes the transformation (whether coarse-graining preserves dynamics), the latter characterizes the network itself (its diffusion properties). 
We argue that both are needed for a complete understanding of the effects of Laplacian renormalization.

%% ====================================================================
%%  HIGHER ORDER EXTENSION
%% ====================================================================

\begin{figure*}
    \centering
    \includegraphics[width=\linewidth]{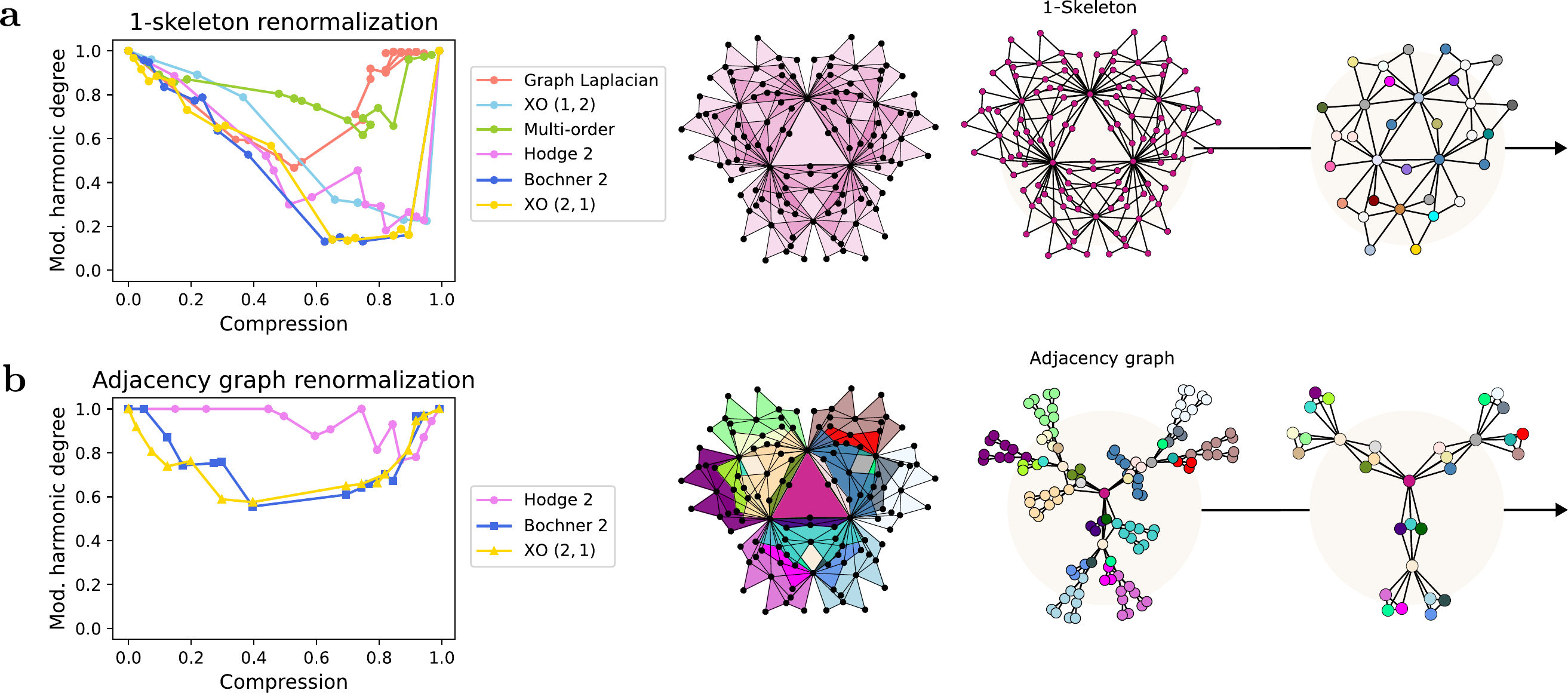}
    \caption{\textbf{Harmonic degrees of Pseudofractal simplicial complex of dimension $2$.} \textbf{a.} Harmonic degree curves for the Laplacian renormalization of the 1-skeleton of a pseudofractal simplicial complex, using different Laplacians. \textbf{b.} Harmonic degree curves evaluated for the Laplacian renormalization of the triangle-edge adjacency graph, using different Laplacians.}
    \label{fig:higher_order}
\end{figure*}

\section{Higher-order Laplacian renormalization}
\label{section:higher-order}

Another natural question is whether harmonic morphisms extend to higher-order networks~\cite{Battiston2021physics,Bianconi2021higher_order_book}, where dynamics live not on nodes but on simplices.
We explore this direction using simplicial complexes, building on the cross-order (XO-) Laplacian renormalization scheme of Ref.~\cite{nurisso2025higher} and extending it to Hodge and Bochner Laplacians~\cite{Eckmann1944HarmonischeFU,Forman2003BochnersMF,Schaub2020random_walks_SC}  via a generalized renormalization scheme (see Supplementary Material~\ref{sec:higher_order_sm} for details).

The central challenge is that harmonic degree is a node-based diagnostic, while higher-order diffusion operates on $k$-simplices.
Evaluating the 1-skeleton (i.e. the underlying graph) under renormalization schemes driven by edge or triangle diffusion consistently disrupts harmonicity: the partition that simplices induce on nodes does not respect node-level locality.
By contrast, the multi-order Laplacian~\cite{Lucas_2020}, which remains node-based while incorporating higher-order information, achieves performance comparable to standard Laplacian renormalization (see \Cref{fig:higher_order}(a)).
This confirms that the issue is not higher-order information per se, but the mismatch between the domain of diffusion and the domain of evaluation.

To resolve this incompatibility, we evaluate harmonicity on the graph where the higher-order dynamics actually live: the $(k,m)$-adjacency graph of the simplicial complex, describing the relations between $k$-simplices mediated by $m$-simplices.
In this approach, we apply harmonic degree to the coarse-graining of the $(k,m)$-adjacency graph induced by higher-order diffusion, rather than attempting to assess coarse-grainings between simplicial complexes directly.
Since the coarse-grained graphs are not guaranteed to be $(k,m)$-adjacency graphs of simplicial complexes, we refer to them as pseudo-adjacency graphs.
Further details on the distinction between adjacency and parallel adjacency graphs appear in the Supplementary Material~\ref{sec:higher_order_sm}.

Within this framework, we see in \Cref{fig:higher_order}(b) that on a 2-dimensional pseudofractal complex~\cite{Dorogovtsev2002pseudofractal}, 2-diffusion via the Hodge Laplacian produces persistent harmonic morphisms outperforming both the Bochner and XO Laplacians---an effect attributable to the Forman curvature term~\cite{Forman2003BochnersMF}, which acts as a reactive component in Hodge diffusion~\cite{nurisso2025higher}. However, when the same 2-diffusion is performed on a 3-dimensional complex, harmonicity is completely disrupted (see Supplementary Material~\ref{sec:higher_order_sm}, Fig.~\ref{fig:si:pf3d}). This dimensional selectivity suggests that harmonic degree under Hodge diffusion could serve as a diagnostic for the intrinsic topological dimension of simplicial complexes, though a rigorous characterization of this phenomenon remains an open question.
Intuitively, $k$-order Hodge diffusion captures interactions at scale $k$; when the complex has non-trivial structure beyond dimension $k$, higher-dimensional simplices introduce boundary effects that break the multiplicity balance required for harmonicity. The harmonic degree under Hodge diffusion thus acts as a resonance probe: high values signal that the chosen diffusion order matches the topological complexity of the data.

%% ====================================================================
%%  DISCUSSION
%% ====================================================================
\section{Discussion}
\label{section:discussion}

\subsection{Discrete conformal invariance}

Our framework establishes harmonic morphisms as discrete analogs of horizontally conformal maps, which in the continuum preserve Brownian motion up to time reparametrization \cite{Lawler2005conformal}. In differential geometry, horizontally conformal (or weakly conformal) maps preserve angles on the subspace orthogonal to the kernel of the differential~\cite{Petersen_diff}, and conformal invariance plays a central role in physics, from conformal field theory~\cite{Introduction_to_Conformal_Field_Theory} to the proof of conformal invariance for lattice models~\cite{Conformal_Invariance_Lattice_Models_Smirnov,Discrete_Complex_Analysis}.

For networks, which are both irregular and non-planar, classical conformal geometry is inapplicable. Yet harmonic morphisms constitute a discrete analog of horizontally conformal transformations, and our work establishes a link between preserving locally harmonic functions and globally mapping random walks under a suitable time change---mirroring how conformal invariance in grids and continua connects local geometry to global properties of physical models.
Our purely combinatorial approach, based only on adjacency and multiplicities, applies to arbitrary networks without geometric embeddings, extending conformal invariance principles beyond lattices and weighted graphs~\cite{Bobenko_Conformal_2015,conformallycovariantoperators}. 
More details on connections to conformal field theory appear in Supplementary Material~\ref{si:link-to-conformal-field-theory}.

\subsection{Connections to statistical mechanics}

Harmonic morphisms also impose a global spectral-arithmetic 
constraint: the Baker-Norine result~\cite{baker_Hyperelliptic_graphs} 
states that if a harmonic morphism $\varphi: G \to G'$ exists, then 
$\kappa_{G'}$ divides $\kappa_G$, where $\kappa$ denotes the 
spanning tree count. 
Through the Kirchhoff Matrix Tree Theorem 
($\kappa_G = \prod_{i=2}^N \lambda_i / N$), this becomes a 
divisibility condition on the pseudodeterminant of the Laplacian, 
complementing the local combinatorial symmetry of horizontal 
conformality with a global spectral invariant.

This constraint highlights a contrast with GNN-based renormalization. 
Harmonic morphisms constrain the Laplacian through a 
pseudodeterminant-level arithmetic relation inherited from 
spanning-tree divisibility, whereas GNN 
renormalization~\cite{Coarse_graining_network_De_Domenico} is designed 
to preserve the heat-trace partition function 
$Z(t) = \mathrm{Tr}(e^{-tL})$, through an objective based on 
preserving $Z(t)/N \approx Z'(t)/N'$ for all $t \geq 0$. 
The two  criteria probe different aspects of network structure: the former is 
tied to exact random-walk projection and a global spanning-tree 
invariant, while the latter targets spectral statistics aggregated 
over all diffusion modes. 
Whether the spanning-tree divisibility 
constraint carries deeper statistical-mechanical implications---for 
instance through connections to Laplacian-based ensembles---remains 
an open question.

\subsection{Limitations and outlook}
Our framework covers simple random walks (via harmonic morphisms) and lazy random walks (via combinatorial conformality) on undirected, unweighted graphs. 
Extensions to biased diffusion, directed networks (where boundary conditions for harmonic functions are more complex), weighted networks (connecting to geometric conformal theories~\cite{Bobenko_Conformal_2015}), and temporal networks remain open problems. 

Finding optimal harmonic morphisms on graphs is a completely unexplored question, related to graph gonality~\cite{Caporaso_Gonality_2}; we propose greedy algorithms (see Supplementary Material Section~\ref{sec:optimization_sm}) that work for $N \leq 60$ but encounter local optima for larger systems, suggesting the need for more principled approaches. 

Finally, a key open question is which networks admit exact harmonic morphism renormalization sequences. 
Our empirical identification of such networks (Facebook, Web-edu, CS Collab) suggests this may be related to the presence of balanced bridge structure, but a graph-theoretic characterization, perhaps connected to the above mentioned gonality~\cite{Caporaso_Gonality_2}, remains to be established. 
More broadly, extending harmonic morphism analysis to other dynamics (synchronization, spreading, consensus~\cite{Consensus_coarse_graining_Barahona}) would require identifying the appropriate ``harmonic'' structure for each operator, but the philosophy---identifying necessary and sufficient conditions for dynamical preservation---applies generally.

\subsection{Conclusion}

This work establishes harmonic morphisms as a principled foundation for network renormalization by identifying them as the exact coarse-graining maps that preserve simple-random-walk transition structure across macro-sets, up to an appropriate time change, bridging discrete probability theory, algebraic graph theory, and statistical mechanics. 
By proving that harmonic morphisms provide a necessary and sufficient condition for first-exit random-walk preservation (Theorem~\ref{theo:RW_preservation}), we answer: \textit{what must network coarse-graining preserve to maintain dynamical equivalence?}
Our key empirical discoveries---distinct dynamical fingerprints for renormalization methods, exact harmonic morphisms spontaneously emerging in real networks, and the complementarity of harmonic degree with entropic susceptibility---demonstrate that the framework provides both theoretical foundations and practical diagnostic tools for multi-scale network analysis.

We view this as one step toward a comprehensive theory of network organization across scales, grounded in the observation that \emph{how systems look at different scales} is determined by \emph{what transformations across scales preserve}.

\begin{acknowledgments}
G. P. is supported by the European Research Council (ERC) Consolidator Grant (Grant Agreement No. 101171380, project RUNES), and MSCA  ``BeyondTheEdge: Higher-Order Networks and Dynamic'' (Grant Agreement No. 101120085) 
\end{acknowledgments}

\subsection*{Author contributions.}

F.M.G.\ and G.P.\ conceived the project. F.M.G.\ developed the theoretical framework and proofs, and wrote the implementation code. F.M.G.\ and M.N.\ performed the numerical experiments and produced the figures. F.G.\ and A.A.\ contributed to the analysis and interpretation of results. G.P.\ supervised and led the project. All authors contributed to writing the manuscript.

\subsection*{Data and code availability.}
Code and data to reproduce the results of this paper are available at \url{https://github.com/nplresearch/harmonic_degree}.

\putbib[biblio]
\end{bibunit}

%% ====================================================================
%%  SUPPLEMENTARY MATERIAL
%% ====================================================================
\begin{bibunit}
\clearpage
\onecolumngrid
\setcounter{figure}{0}
\setcounter{table}{0}
\setcounter{equation}{0}
\setcounter{page}{1}
\setcounter{section}{0}

\makeatletter
\renewcommand{\thesection}{S\arabic{section}}
\renewcommand{\thesubsection}{\thesection.\arabic{subsection}}
\renewcommand{\thesubsubsection}{\thesubsection.\arabic{subsubsection}}
\renewcommand{\thefigure}{S\arabic{figure}}
\renewcommand{\theequation}{S\arabic{equation}}
\renewcommand{\thetable}{S\arabic{table}}

\begin{center}
	\textbf{\large Supplementary Material}
\end{center}

%% ====================================================================
\section{Proofs}
\label{sec:proofs_sm}

Throughout this section, we use the notation established in the main text: $G=(V,E)$ and $\mathcal{G}=(\mathcal{V},\mathcal{E})$ are two graphs, $\varphi: V \to \mathcal{V}$ is surjective, $G_y$ is the subgraph induced on $\varphi^{-1}(y)$, $\partial G_y$ is its boundary, $\partial G_{yy'} = \partial G_y \cap \varphi^{-1}(y')$, and $\mathcal{H}_y(x, y') = \mathcal{H}_{G_y}(x, \partial G_{yy'})$.

We recall that $\mathcal{H}(\cdot, B)$ is the unique harmonic function on $G$ satisfying the discrete Laplace equation $\Delta u = 0$ with boundary conditions $u = 1$ on $B$ and $u = 0$ on $\partial G \setminus B$.

We first establish a useful lemma.

\begin{lemma}
\label{lem:nonempty_boundary}
Let $\varphi$ be a surjective harmonic morphism and $G$ connected. If $y' \sim y$ then $\partial G_{yy'} \neq \emptyset$, excluding the trivial case $\partial G_y = \emptyset$ (i.e., $\mathcal{G}$ is a single node).
\end{lemma}

\begin{proof}
Assume by contradiction that $y' \sim y$ but $\partial G_{yy'} = \emptyset$ and $\partial G_y \neq \emptyset$. Define $f: \mathcal{V} \to \mathbb{R}$ by $f(y) = 1$, $f(y') = -1$, and $f = 0$ elsewhere; this function is harmonic at $y$. Since $\partial G_y \neq \emptyset$, there exists $x' \in G_y$ adjacent to at least one node in $\partial G_y$. Now $f \circ \varphi$ equals $1$ on $\varphi^{-1}(y)$, $-1$ on $\varphi^{-1}(y')$, and $0$ elsewhere. Since $x'$ has no neighbors in $\varphi^{-1}(y')$ (by $\partial G_{yy'} = \emptyset$) but does have neighbors in $\partial G_y$, the pullback $f \circ \varphi$ is not harmonic at $x'$, leading to a contradiction.
\end{proof}

\subsection{Proof of Theorem~\ref*{theo:RW_preservation} (Random walk preservation)}

\textit{Statement.} $\varphi$ is a harmonic morphism if and only if $\mathcal{H}_y(x, y') = 1/\deg(y)$ for all $y \in \mathcal{V}$, $x \in \varphi^{-1}(y)$, and $y' \sim y$.

\begin{proof}

\noindent ($\Rightarrow$) Assume $\varphi$ is a harmonic morphism. If $\tau = \infty$ the statement holds vacuously. For finite graphs (our setting), $\tau < \infty$ almost surely. We condition on the last position before exit: let $X^x_{\tau-1} = z$. By Markovianity and horizontal conformality (condition~(ii) of Definition~\ref{def:horizontal_conformal}),
\[
\mathbb{P}(X^x_\tau \in \partial G_{yy'} \mid X^x_{\tau-1} = z) = \mathbb{P}(X^z_\tau \in \partial G_{yy'} \mid \tau = 1) = \frac{1}{\deg(y)}.
\]
Averaging over all possible $z$:
\[
\mathcal{H}_y(x, y') = \sum_{z \in G_y} \mathbb{P}(X^x_\tau \in \partial G_{yy'} \mid X^x_{\tau-1} = z)\, \mathbb{P}(X^x_{\tau-1} = z) = \frac{1}{\deg(y)}.
\]

\noindent ($\Leftarrow$) Assume the harmonic measure condition holds. We first show $\{y' : \partial G_{yy'} \neq \emptyset\} = \mathcal{N}(y)$, where $\mathcal{N}(y)$ denotes the set of vertices adjacent to $y$.

If $y' \sim y$ then $\partial G_{yy'} \neq \emptyset$, since otherwise $\mathcal{H}_y(x, y') = 0 \neq \mathbb{P}(Y_{t=1} = y' \mid Y_{t=0} = y)$. Hence $\mathcal{N}(y) \subseteq \{y' : \partial G_{yy'} \neq \emptyset\}$.

Now suppose by contradiction that $\mathcal{N}(y) \subsetneq \{y' : \partial G_{yy'} \neq \emptyset\}$. For any $x \in V$,
\[
1 = \sum_{y' \sim y} \mathbb{P}(Y_{t=1} = y' \mid Y_{t=0} = y) < \sum_{\{y' : \partial G_{yy'} \neq \emptyset\}} \mathcal{H}_y(x, y') = 1,
\]
using $\mathcal{N}(y) \subseteq \{y' : \partial G_{yy'} \neq \emptyset\}$, the hypothesis, and the strict inclusion---a contradiction. Thus $\{y' : \partial G_{yy'} \neq \emptyset\} = \mathcal{N}(y)$, which implies $\varphi$ is a weak homomorphism: if $x \sim z$ with $\varphi(x) \neq \varphi(z)$ and $\varphi(x) \not\sim \varphi(z)$, we would have $\partial G_{\varphi(x)\,\varphi(z)} \neq \emptyset$ but $\varphi(z) \notin \mathcal{N}(\varphi(x))$, contradicting the equality just established.

Now fix $y \in \mathcal{V}$ and $y' \sim y$. The function $\mathcal{H}_y(\cdot, y')$ is harmonic on $G_y$ with boundary values $1$ on $\partial G_{yy'}$ and $0$ on $\partial G_y \setminus \partial G_{yy'}$. Take $x \in \varphi^{-1}(y)$ and let $k_{y'}$ denote the number of neighbors of $x$ in $\varphi^{-1}(y')$, and $k_y$ the number in $\varphi^{-1}(y)$. By harmonicity of $\mathcal{H}_y(\cdot, y')$:
\[
(\deg(x) - k_y)\,\mathcal{H}_y(x, y') = k_{y'}.
\]
If $k_y = \deg(x)$ (all neighbors of $x$ are in the same macro-set), then $k_{y'} = 0$ for all $y' \sim y$ and the condition holds trivially. Otherwise, $\mathcal{H}_y(x, y') = k_{y'}/(\deg(x) - k_y) = 1/\deg(y)$ by hypothesis, so $k_{y'}$ is constant across all $y' \sim y$. This is exactly horizontal conformality.
\end{proof}

\subsection{Proof of Theorem~\ref*{thm:lazy_rw} (Lazy random walk preservation)}

\textit{Statement.} $\varphi$ is combinatorial conformal if and only if $\mathbb{P}(\mathcal{L}_{t=1} \in \varphi^{-1}(y') \mid \mathcal{L}_{t=0} = x) = 1/(\deg(y)+1)$ for all $y'$ with $y' \sim y$ or $y' = y$.

\begin{proof}

\noindent ($\Rightarrow$) Fix $y \in \mathcal{V}$ and $x \in \varphi^{-1}(y)$. For $y' \sim y$, let $k_{y'} = |\{z \in \varphi^{-1}(y') : z \sim x\}|$; for $y' = y$, let $k_y = |\{z \in \varphi^{-1}(y) : z \sim x \text{ or } z = x\}|$. Since $\varphi$ is a weak homomorphism, these induce a partition of $\mathcal{N}(x) \cup \{x\}$. By combinatorial conformality, all $k_{y'}$ are equal. Since $\mathbb{P}(\mathcal{L}_{t=1} \in \varphi^{-1}(y') \mid \mathcal{L}_{t=0} = x) = k_{y'}/(\deg(x)+1)$ and the pre-images partition $\mathcal{N}(x) \cup \{x\}$, the probabilities sum to~1, giving $k_{y'}/(\deg(x)+1) = 1/(\deg(y)+1)$ for each $y'$.

\noindent ($\Leftarrow$) First, we show $\varphi$ is a weak homomorphism. Suppose $\mathcal{N}(y) \cup \{y\} \subsetneq \{y' : \mathbb{P}(\mathcal{L}_{t=1} \in \varphi^{-1}(y') \mid \mathcal{L}_{t=0} = x) \neq 0\}$. Then $1 = \sum_{y' \sim y \text{ or } y'=y} 1/(\deg(y)+1) < \sum_{\{y': \mathbb{P} \neq 0\}} \mathbb{P}(\mathcal{L}_{t=1} \in \varphi^{-1}(y') \mid \mathcal{L}_{t=0} = x) = 1$, a contradiction. Hence $\varphi$ is a weak homomorphism, and the pre-images of $\{y' : y' \sim y \text{ or } y' = y\}$ partition $\mathcal{N}(x) \cup \{x\}$.

Now, $\mathbb{P}(\mathcal{L}_{t=1} \in \varphi^{-1}(y') \mid \mathcal{L}_{t=0} = x) = k_{y'}/(\deg(x)+1) = 1/(\deg(y)+1)$ by hypothesis, so all $k_{y'}$ are equal---which is exactly the definition of combinatorial conformality.
\end{proof}

%% ====================================================================
\section{Conformal degree metrics}
\label{sec:conformal_metrics_sm}

Analogously to the harmonic degree metrics defined in the main text (Sec.~\ref*{section:harmonic_degree}), we define conformal degree metrics based on combinatorial conformality (Definition~\ref*{def:combinatorial_conformal}). For each node $x$ with $\varphi(x)=y$, we compute the extended multiplicity
\begin{equation}
\tilde{k}_{y'}(x) := |\{ z \in \varphi^{-1}(y') : z \sim x \text{ or } z = x \}|
\end{equation}
for all $y'$ with $y' \sim y$ or $y' = y$. A node $x$ is a \emph{conformal node} if all $\tilde{k}_{y'}(x)$ are equal, i.e., the multiplicity condition extends to include the node's own macro-set (counting itself). We denote by $CF(A)$ the subset of conformal nodes in $A \subseteq V$.

The three conformal degree metrics are:
\begin{enumerate}
\item \textbf{Mean conformal degree:}
\begin{equation}
CF_{\text{mean}} = \frac{|CF(V)|}{|V|},
\end{equation}
the fraction of conformal nodes in the network.

\item \textbf{Modified conformal degree:}
\begin{equation}
CF_{\text{mod}} = \frac{1}{|\mathcal{V}|}\sum_{y \in \mathcal{V}} \frac{|CF(\varphi^{-1}(y))|}{|\varphi^{-1}(y)|},
\end{equation}
averaging conformality within each macro-set to prevent large clusters from dominating.

\item \textbf{Conformal deviation:}
\begin{equation}
CF_{\text{Dev}} = \frac{1}{|V|}\sum_{x \in V} \mathrm{std}\!\left(\{\tilde{k}_{y'}(x) : y' \sim \varphi(x) \text{ or } y' = \varphi(x)\}\right),
\end{equation}
a continuous measure of multiplicity imbalance including the within-macro-set direction.
\end{enumerate}

In practice, combinatorial conformality is far more restrictive than horizontal conformality. Across 50 networks $\times$ 5 clustering methods = 250 experiments (Sec.~\ref{sec:clustering_sm}), $CF_{\text{mod}}\simeq 0$ in approximately 85\% of cases. The most notable exception is Infomap on the Facebook network ($CF_{\text{mod}} = 0.625$), where high-betweenness nodes form singleton clusters that trivially satisfy conformality through the identity mapping. Non-zero conformal degree generally signals nodes mapped to themselves rather than genuine balanced mixing including within-macro-set connections.

%% ====================================================================
\section{Computational implementation of harmonic degrees}
\label{sec:computational_sm}

Computing harmonic and conformal degrees is straightforward for a given partition. The algorithm proceeds as follows:

\begin{enumerate}
\item \textbf{Construct induced graph.} Given partition $\varphi^{-1}(y)$ for each $y \in \mathcal{V}$, create $\mathcal{G}$ with edge $y \sim y'$ whenever $\exists x \in \varphi^{-1}(y), z \in \varphi^{-1}(y')$ with $x \sim z$ in $G$.

\item \textbf{For each node $x \in V$:}
\begin{itemize}
\item Determine $y = \varphi(x)$
\item For each neighboring macro-node $y' \sim y$, count $k_{y'}(x) = |\{z \in \varphi^{-1}(y') : z \sim x\}|$
\item Check if all $k_{y'}(x)$ are equal (harmonic node)
\item Compute $\sigma_x = \text{std}(\{k_{y'}(x)\})$
\item For conformality: include $\tilde{k}_y(x) = |\{z \in \varphi^{-1}(y) : z \sim x \text{ or } z=x\}|$
\end{itemize}

\item \textbf{Aggregate metrics.} Compute $H_{\text{mean}}$, $H_{\text{mod}}$, $H_{\text{Dev}}$ and conformal analogs.
\end{enumerate}

%The computational complexity is $O(|E| + |V||\mathcal{V}|)$: we iterate over each edge once to build the induced graph, then for each node we check $|\mathcal{N}(\varphi(x))|$ neighboring macro-nodes. For renormalization sequences producing $\log N$ scales, total cost is $O(|E| \log N)$, efficiently computable even for large networks.

For Laplacian renormalization, we implement the merging criterion as $\rho_{ij}(t) \geq \min(\rho_{ii}(t), \rho_{jj}(t)) - \varepsilon$, where $\varepsilon$ is a small regularization parameter ($\varepsilon \in \{0, 10^{-3}, 10^{-4}\}$ depending on the network) introduced to control numerical stability near the merging threshold. This modification does not alter the qualitative behavior of the renormalization flow.

%% ====================================================================
\section{Community detection as one-step coarse-graining}
\label{sec:clustering_sm}

Community detection algorithms partition networks into groups, providing a natural testbed for our metrics. Any partition induces a weak homomorphism (Definition~\ref*{def:horizontal_conformal}, condition~i) to a coarse-grained graph: macro-nodes represent groups, and $\varphi(x) \sim \varphi(z)$ whenever $x \sim z$ with $\varphi(x) \neq \varphi(z)$. We can then compute harmonic and conformal degrees for the resulting coarse-graining.

\textbf{What does harmonic degree measure for clustering?} It evaluates clustering from a global, ``external'' perspective. While community detection algorithms identify nodes with similar internal properties (high intra-group density, shared roles), harmonic degrees measure the \emph{mixing symmetries} induced at the macro-scale: are macro-sets connected in a balanced way that preserves random walk structure?

High harmonic degree indicates that macro-set boundaries are organized such that walkers exit uniformly toward neighboring macro-sets, matching coarse-grained transition probabilities. Low harmonic degree reveals asymmetric mixing: some nodes have many connections to one neighboring macro-set but few to another, distorting the random walk when projected.

Importantly, harmonic degree is \emph{not} a quality metric for clustering in the traditional sense. It does not assess intra-group coherence or separation (modularity, conductance). Instead, it quantifies whether the resulting partition is compatible with random walk renormalization. A partition with perfect modularity could have low harmonicity if macro-sets connect asymmetrically; conversely, a partition with moderate modularity could achieve high harmonicity through balanced mixing.

\subsection{Benchmark datasets and clustering methods}
We consider approximately 50 real networks, divided into 6 categories: social, infrastructure, biological, tech/web, animal, and collaboration.
Data sources include SNAP~\cite{snapnets}, BioSNAP~\cite{biosnapnets}, and the Network Repository~\cite{Network_Repository}.
Table~\ref{tab:dataset_sources} lists the networks discussed individually in this work together with their original references.
For each network, we perform community detection using five established algorithms: label propagation~\cite{Lab_prob,Lab_prop_networks}, greedy modularity~\cite{Modularity_Newman_2006}, Louvain algorithm~\cite{Louvain_algo}, Infomap~\cite{mapequation2025software,Map_equation_networks}, and non-parametric Bayesian inference via Stochastic Block Models~\cite{Peixoto_Bayesian_SBM_2019}.
The same networks are also used in the renormalization analysis. The complete list of networks and per-network results across all methods are available in the accompanying code repository.
 
\begin{table}[h]
\centering
\caption{Networks discussed individually in this work, with original data references.
Networks obtained from aggregate repositories (SNAP, BioSNAP, Network Repository) are cited to the original study that produced the data.}
\label{tab:dataset_sources}
\begin{tabular}{lll}
\toprule
Network & Category & Reference \\
\midrule
Facebook (Caltech36)          & Social          & Traud et al.~\cite{traud2012social} \\
Facebook (Haverford76)        & Social          & Traud et al.~\cite{traud2012social} \\
%% ADD/REMOVE Facebook subnetwork rows as needed (Colgate88, UC64, etc.)
Song of Ice and Fire          & Social          & Beveridge \& Shan~\cite{beveridge2016network} \\
Euro-Road                     & Infrastructure  & \v{S}ubelj \& Bajec~\cite{subelj2011euroroad} \\
Minnesota Road                & Infrastructure  & Bader et al.~\cite{bader2013dimacs10} \\  %% VERIFY: confirm this is the DIMACS10 version
Power: bcspwr09               & Infrastructure  & Duff et al.~\cite{duff1989sparse} \\
NetSci collaboration          & Collaboration   & Newman~\cite{newman2006finding} \\
CS collaboration              & Collaboration   & Newman~\cite{newman2006finding} \\  %% VERIFY: confirm same source or different
\textit{C.~elegans} metabolic & Biological      & Duch \& Arenas~\cite{duch2005community} \\
Yeast (protein interaction)   & Biological      & Jeong et al.~\cite{jeong2001lethality} \\ %% VERIFY: confirm which PPI version
Web: education                & Tech/Web        & Gleich \& Rossi~\cite{Network_Repository} \\ %% VERIFY: identify original source
%% Weaver                     & Animal          & %% VERIFY: identify original study from ASNR
%% Tortoise                   & Animal          & %% VERIFY: identify original study from ASNR
%% Bio: Plant                 & Biological      & %% VERIFY: identify original source
\bottomrule
\end{tabular}
\end{table}

\subsection{No universal best method}
Harmonic degree values vary substantially across both networks and clustering methods, with no algorithm consistently producing the highest or lowest harmonicity. This network-dependence suggests harmonic degree captures genuine structural variation rather than merely reflecting algorithmic biases.

Table~\ref{tab:example_networks} shows representative examples selected to illustrate the range of behaviors; full per-network results are provided in the code repository. 
The Facebook network achieves high harmonicity ($H_{\text{mod}} > 0.96$) across label propagation, greedy modularity, Louvain, and Infomap, but SBM produces very low values ($H_{\text{mod}} = 0.038$). Conversely, the Power: bcspwr09 network shows SBM achieving highest harmonicity ($H_{\text{mod}} = 0.951$) while label propagation produces moderate values ($H_{\text{mod}} = 0.497$). The C.~Elegans metabolic network presents yet another pattern: label propagation achieves $H_{\text{mod}} = 0.889$ while all other methods yield $H_{\text{mod}} < 0.4$. However, in this network label propagation produces a spurious macro-cluster containing the majority of nodes, increasing the harmonicity.

\begin{table}[h]
\centering
\caption{Harmonic and conformal degrees for selected networks across clustering methods.}
\label{tab:example_networks}
\begin{tabular}{llccc}
\toprule
Network & Method & $H_{\text{mod}}$ & $H_{\text{Dev}}$ & $CF_{\text{mod}}$ \\
\midrule
Facebook & LabProp & 0.966 & 0.022 & 0.000 \\
& Greedy Mod. & 0.989 & 0.008 & 0.000 \\
& Louvain & 0.989 & 0.008 & 0.000 \\
& Infomap & 0.985 & 0.010 & 0.625 \\
& SBM & 0.038 & 0.962 & 0.038 \\
\midrule
Web: edu & LabProp & 0.960 & 0.005 & 0.000 \\
& Greedy Mod. & 0.999 & 0.002 & 0.000 \\
& Louvain & 0.999 & 0.002 & 0.000 \\
& Infomap & 0.977 & 0.010 & 0.000 \\
& SBM & 0.831 & 0.116 & 0.000 \\
\midrule
C.~Elegans & LabProp & 0.889 & 0.030 & 0.000 \\
& Greedy Mod. & 0.382 & 0.564 & 0.000 \\
& Louvain & 0.364 & 0.619 & 0.000 \\
& Infomap & 0.358 & 0.510 & 0.000 \\
& SBM & 0.061 & 0.896 & 0.000 \\
\midrule
Power BCS09 & LabProp & 0.497 & 0.233 & 0.000 \\
& Greedy Mod. & 0.890 & 0.051 & 0.000 \\
& Louvain & 0.876 & 0.058 & 0.000 \\
& Infomap & 0.646 & 0.155 & 0.000 \\
& SBM & 0.951 & 0.023 & 0.000 \\
\midrule
CS Collab & LabProp & 0.910 & 0.044 & 0.000 \\
& Greedy Mod. & 0.964 & 0.017 & 0.000 \\
& Louvain & 0.960 & 0.018 & 0.000 \\
& Infomap & 0.919 & 0.032 & 0.000 \\
& SBM & 0.202 & 0.653 & 0.130 \\
\midrule
Facebook Haverford76 & LabProp & 0.84 & 1.16 & 0.000 \\
& Greedy Mod. & 0.018 & 10.1 & 0.000 \\
& Louvain & 0.02 & 6.44 & 0.000 \\
& Infomap & 0.328 & 1.66 & 0.001 \\
& SBM & 0.003 &1.63 & 0.000 \\
\bottomrule
\end{tabular}
\end{table}

% \begin{figure}[H]
%     \centering
%     \includegraphics[width= 1 \linewidth]{figures/Clustering.pdf}
%     \caption{Examples of analyzed networks with varying harmonic behavior across clustering algorithms. Starting from the top left: Facebook network, Web: edu network, C. Elegans metabolic network, Power: bcspwr09 network, Facebook Haverford76 network, and CS Collaboration network.}
%     \label{fig: Fig_3}
% \end{figure}

\subsection{Structural patterns enabling high harmonicity}

High harmonicity generally occurs in networks where macro-sets communicate through a small number of well-defined interface nodes. Hub-spoke architectures (e.g., Barab\'asi-Albert networks) naturally achieve high harmonicity through node collapses. More generally, high harmonicity requires that boundary nodes distribute their external connections uniformly across neighboring macro-sets. Some collaboration networks (CS, NetSci) consistently show $H_{\text{mod}} > 0.9$ across most methods, reflecting natural group structure with distributed inter-group ties. Some Web networks (education, Indochina) achieve mean $H_{\text{mod}} > 0.80$, as link structure organized around topics creates macro-sets with relatively uniform external connectivity.
Conversely, low harmonicity indicates asymmetric mixing. 
Some Facebook subnetworks (Colgate88, UC64, Haverford76) show $H_{\text{mod}} < 0.2$ for modularity-based methods, reflecting heterogeneous degree distributions within macro-sets and irregular boundary structures.

\subsection{Method-specific patterns}

\textbf{Greedy modularity and Louvain} produce similar results (shared modularity objective), achieving high harmonicity on collaboration (mean $H_{\text{mod}} = 0.88$ and $0.865$ respectively), animal (mean $H_{\text{mod}} = 0.83$ and $0.81$), and infrastructure networks (mean $H_{\text{mod}} = 0.89$ and $0.88$), but performing poorly on biological networks (mean $H_{\text{mod}} = 0.32$ and $0.27$).

\textbf{Label propagation} shows high variance, sometimes producing giant clusters with small satellites yielding very high harmonicity, especially on biological networks (mean $H_{\text{mod}} > 0.90$). It gives lower values on infrastructure networks (mean $H_{\text{mod}} = 0.47$) and mean $H_{\text{mod}} \geq 0.7$ for all other categories (social, collaboration, animal, tech).

\textbf{Infomap} typically produces mid-range harmonicity (mean $H_{\text{mod}} = 0.51$) but consistently achieves non-zero conformal degrees, even moderately high in some cases such as the Facebook network ($CF_{\text{mod}} = 0.625$). This reflects Infomap's focus on random walk flow: high-betweenness nodes forming singleton clusters perceive the identity map, satisfying conformality trivially.

\textbf{Stochastic Block Model} exhibits strong network-dependence: highest harmonicity on infrastructure networks (total mean $H_{\text{mod}} = 0.77$), especially for power grids, but very low values on social and protein networks (mean $H_{\text{mod}} < 0.2$). SBM identifies statistically significant block structure that may not align with mixing symmetries.

\subsection{Systematic patterns across network categories}

\textbf{Social networks} ($n=12$): High variability. Facebook friendship networks range from $H_{\text{mod}} = 0.989$ (main Facebook) to $H_{\text{mod}} = 0.176$ (Haverford76) under greedy modularity. SBM is generally low across all these networks.

\textbf{Infrastructure networks} ($n=7$): Consistently high harmonicity, with a mean $H_{\text{mod}} = 0.74$ across all methods. Physical constraints (geographic distribution, load balancing) likely produce hierarchical structure with symmetric communication patterns.

\textbf{Biological networks} ($n=16$): Very high variability ranging from mean $H_{\text{mod}} = 0.30$ in the neuron protein-protein interaction network to mean $H_{\text{mod}} = 0.81$ in the plant metabolic network. Label propagation often produces giant components yielding spuriously high harmonicity through collapse rather than genuine balanced mixing.

\textbf{Tech/Web networks} ($n=9$): Moderately-high harmonicity with total mean $H_{\text{mod}} = 0.64$.
Web structure organized around topics creates relatively balanced cross-group connectivity.

\textbf{Animal networks} ($n=2$): Moderate to high values, total mean $H_{\text{mod}} = 0.75$.

\textbf{Collaboration networks} ($n=4$): Consistently high across most methods; for example, CS has $H_{\text{mod}} > 0.90$ with all methods except SBM ($H_{\text{mod}} = 0.20$), and NetSci has $H_{\text{mod}} = 0.82$ with all methods. The total mean $H_{\text{mod}}$ is $0.75$.

\subsection{Conformality: singleton clusters and identity mappings}

Across 50 networks $\times$ 5 methods = 250 experiments, combinatorial conformality is predominantly zero ($CF_{\text{mean}} = CF_{\text{mod}} = 0$). The most notable case is Infomap on the Facebook network, with $CF_{\text{mod}} = 0.625$. Infomap consistently yields non-zero conformality, as do the greedy optimization methods, producing non-null conformalities in infrastructure, tech, and social networks, as well as a few biological networks.
Non-zero conformal degree is often a sign that the coarse-graining includes nodes mapped to themselves (identity on a subset). Since conformality requires lazy random walk preservation (Theorem~\ref*{thm:lazy_rw}), these cases involve no temporal scale change.

\section{Additional renormalization results}
\label{sec:detailed_renorm_sm}
 
This section provides the disaggregated harmonic and conformal degree analyses supporting the main text discussion in Sec.~\ref*{section:RG} and Sec.~\ref*{section:laplacian_deep}. We first present aggregate results across all 16 networks for which all three renormalization methods could be run, then examine individual networks organized by method to illustrate how the dynamical fingerprints described in the main text manifest in specific topologies.
 
\subsection{Aggregate results across networks}
Figures~\ref{fig:si:all_results_harmonic} and~\ref{fig:si:all_results_conformal} show per-network harmonic and conformal degree curves alongside their averages. The three fingerprints identified in the main text---the S-shaped geometric curve, the high-low-high Laplacian pattern, and the uniformly low GNN signature---are clearly visible in the averaged curves and are broadly reproduced across individual networks, confirming that these are method-level signatures rather than artifacts of particular topologies.

Several features of the disaggregated data merit comment. First, the variance across networks is smallest for Laplacian renormalization at small compression ($\eta < 0.4$), where nearly all networks exhibit $H_{\text{mod}} > 0.7$. This reflects the universal tendency of short-time diffusion to merge only tightly connected node pairs, which naturally produce symmetric macro-set boundaries. The variance increases at intermediate compression, where network-specific heterogeneity in community structure determines how evenly diffusion spreads.

Second, geometric renormalization shows the widest inter-network spread, particularly at intermediate scales. Networks with clear latent geometric structure (road networks, spatial networks) follow the S-curve closely, while networks lacking such structure (character interaction networks, some social networks) can deviate substantially, sometimes failing to reach high harmonicity even at maximal compression. 
 
Third, for GNN-based renormalization, the low harmonicity is remarkably consistent across networks: no network achieves $H_{\text{mod}} > 0.6$ at any compression level, and most drop below $0.2$ by $\eta = 0.5$. This uniformity reflects the method's spectrum-preserving objective, which distributes nodes across macro-sets by structural role irrespective of local mixing symmetry.
 
For conformal degree (Fig.~\ref{fig:si:all_results_conformal}), the dominant pattern is $CF_{\text{mod}} \approx 0$ for geometric and GNN renormalization across all networks. Laplacian renormalization shows non-trivial conformal degree in a subset of networks, typically those that also achieve high harmonic degree, though $CF_{\text{mod}}$ remains well below $H_{\text{mod}}$ in all cases. The gap $H_{\text{mod}} - CF_{\text{mod}}$ reflects the fact that most Laplacian merging operations create macro-sets with asymmetric internal connectivity even when external multiplicities are balanced.
 
\begin{figure}[tbp]
    \centering
    \includegraphics[width=\linewidth]{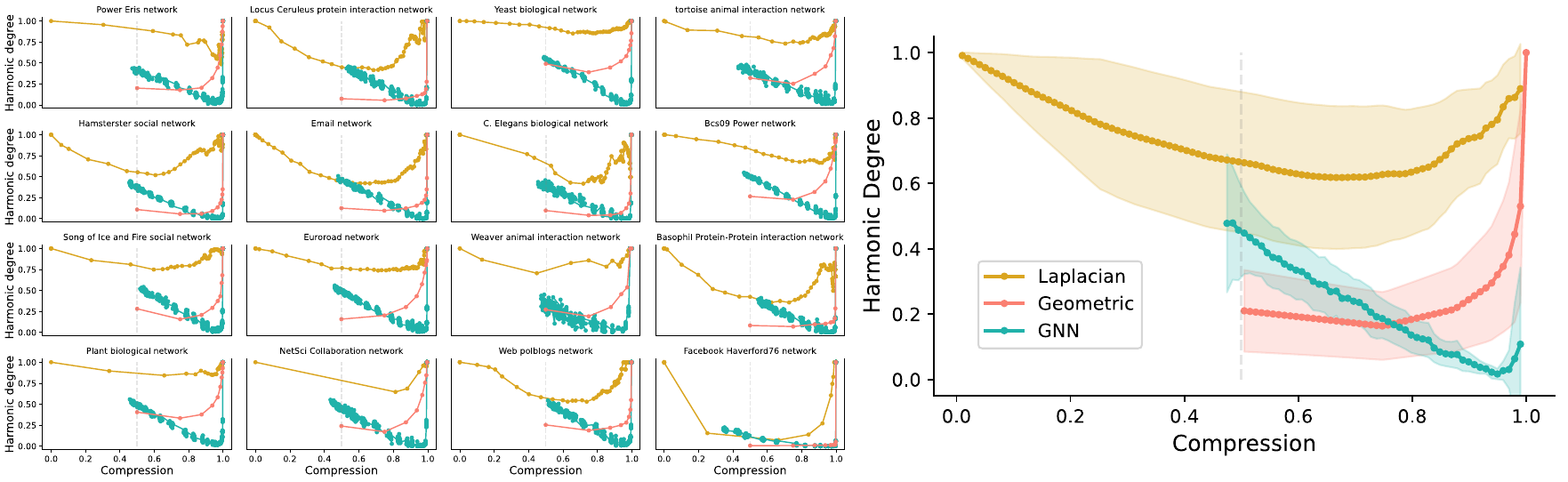}
    \caption{Harmonic degree as a function of compression for 16 networks (left) and their average (right), for geometric (top), Laplacian (middle), and GNN-based (bottom) renormalization.}
    \label{fig:si:all_results_harmonic}
\end{figure}
 
\begin{figure}[tbp]
    \centering
    \includegraphics[width=\linewidth]{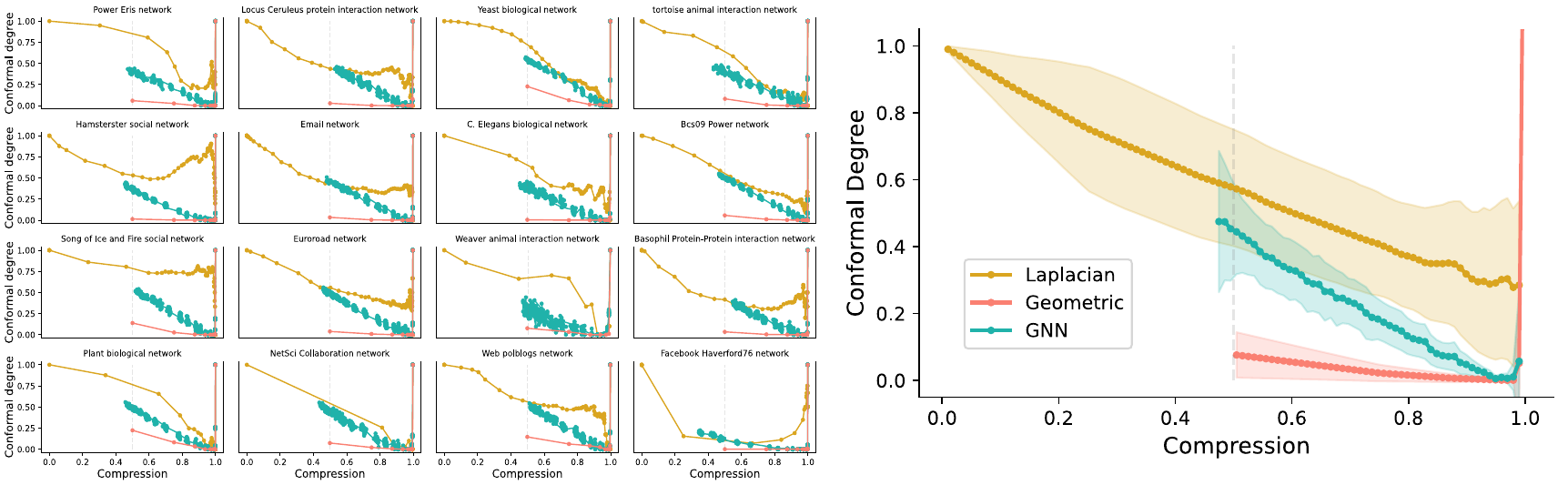}
    \caption{Conformal degree as a function of compression for 16 networks (left) and their average (right), for geometric (top), Laplacian (middle), and GNN-based (bottom) renormalization.}
    \label{fig:si:all_results_conformal}
\end{figure}

\subsection{Geometric renormalization: individual networks}
 
We present three networks that illustrate the range of behavior under geometric renormalization. The Power: bcspwr09 network (Fig.~\ref{fig:si:geometric_renorm_examples}a) is an infrastructure network whose nodes are geographically distributed; its harmonic degree curve follows the canonical S-shape, with low harmonicity at early iterations giving way to high values as the angular merging captures the coarse spatial organization of the power grid. This is the regime where geometric renormalization performs closest to Laplacian methods.
 
The Minnesota Road network (Fig.~\ref{fig:si:geometric_renorm_examples}b) shows a similar S-curve, consistent with its clear planar geographic structure. Road networks are natural candidates for hyperbolic embedding, and the angular proximity used by the Mercator algorithm aligns well with geographic adjacency at coarse scales.
 
The Song of Ice and Fire character network (Fig.~\ref{fig:si:geometric_renorm_examples}c) presents a contrasting case: geometric renormalization fails to achieve high harmonicity at any scale. This network, built from character co-occurrences in a narrative, lacks a natural geometric embedding. The angular proximity criterion groups characters that are ``close'' in the embedding space but may interact with very different subsets of the cast, producing asymmetric multiplicities throughout the RG flow. This example illustrates that the S-curve fingerprint requires genuine latent geometry; without it, the global embedding imposes artificial structure that conflicts with local connectivity at all scales.

\begin{figure}
\centering
\includegraphics[width=.8\linewidth]{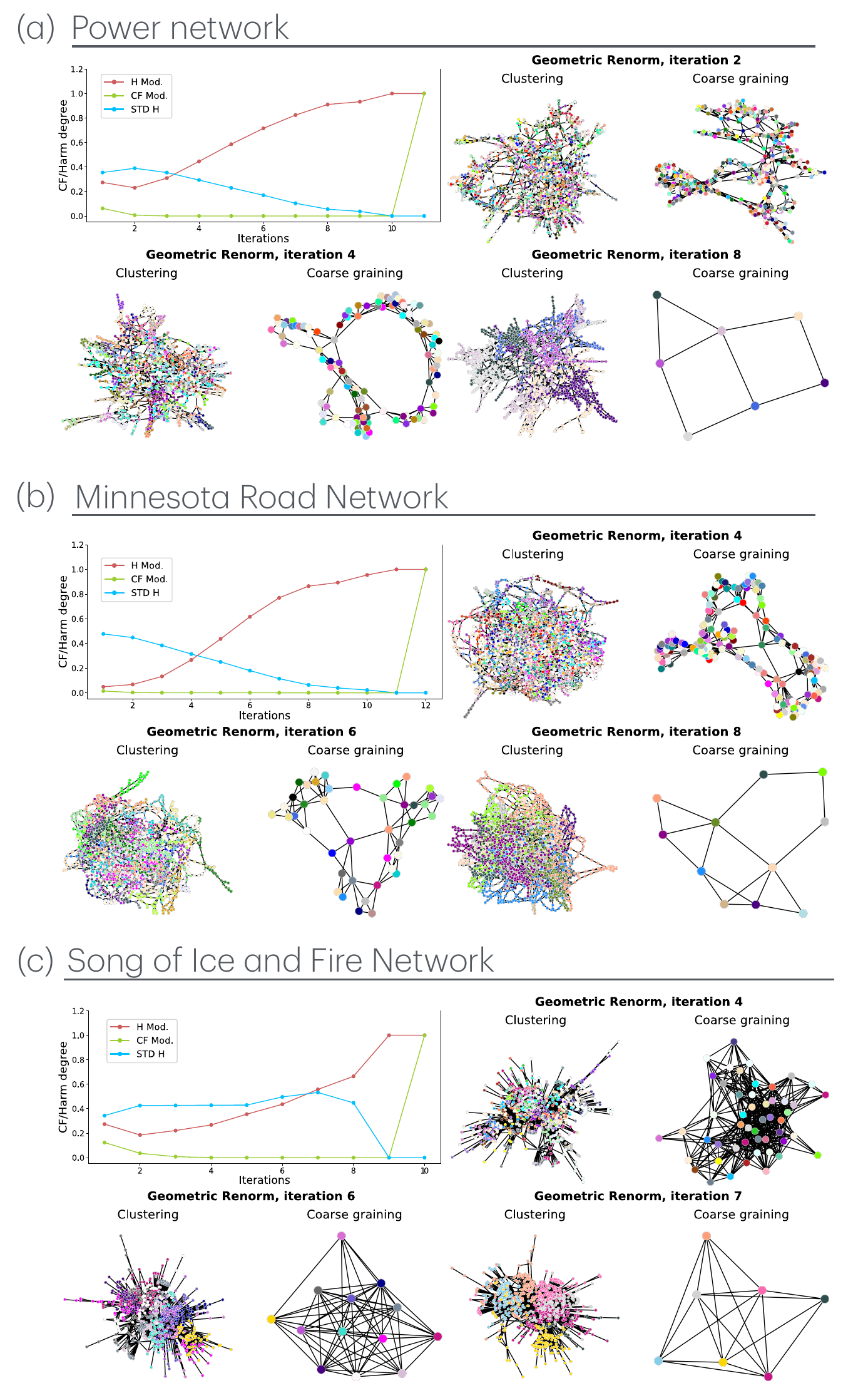}
\caption{Geometric renormalization of the (a) power (bcspwr09) network, (b) Minnesota Road network, and (c) Song of Ice and Fire character network. 
In the first two spatially embedded networks, the harmonic degree follows the canonical S-shaped curve, reflecting underlying geographic organization and increasing harmonicity as angular merging captures coarse spatial structure. By contrast, the Song of Ice and Fire network lacks such latent geometry, and correspondingly does not attain high harmonicity at any scale.}\label{fig:si:geometric_renorm_examples}
\end{figure}
 
\subsection{Laplacian renormalization: individual networks}

\begin{figure}
    \centering
    \includegraphics[width=.9\linewidth]{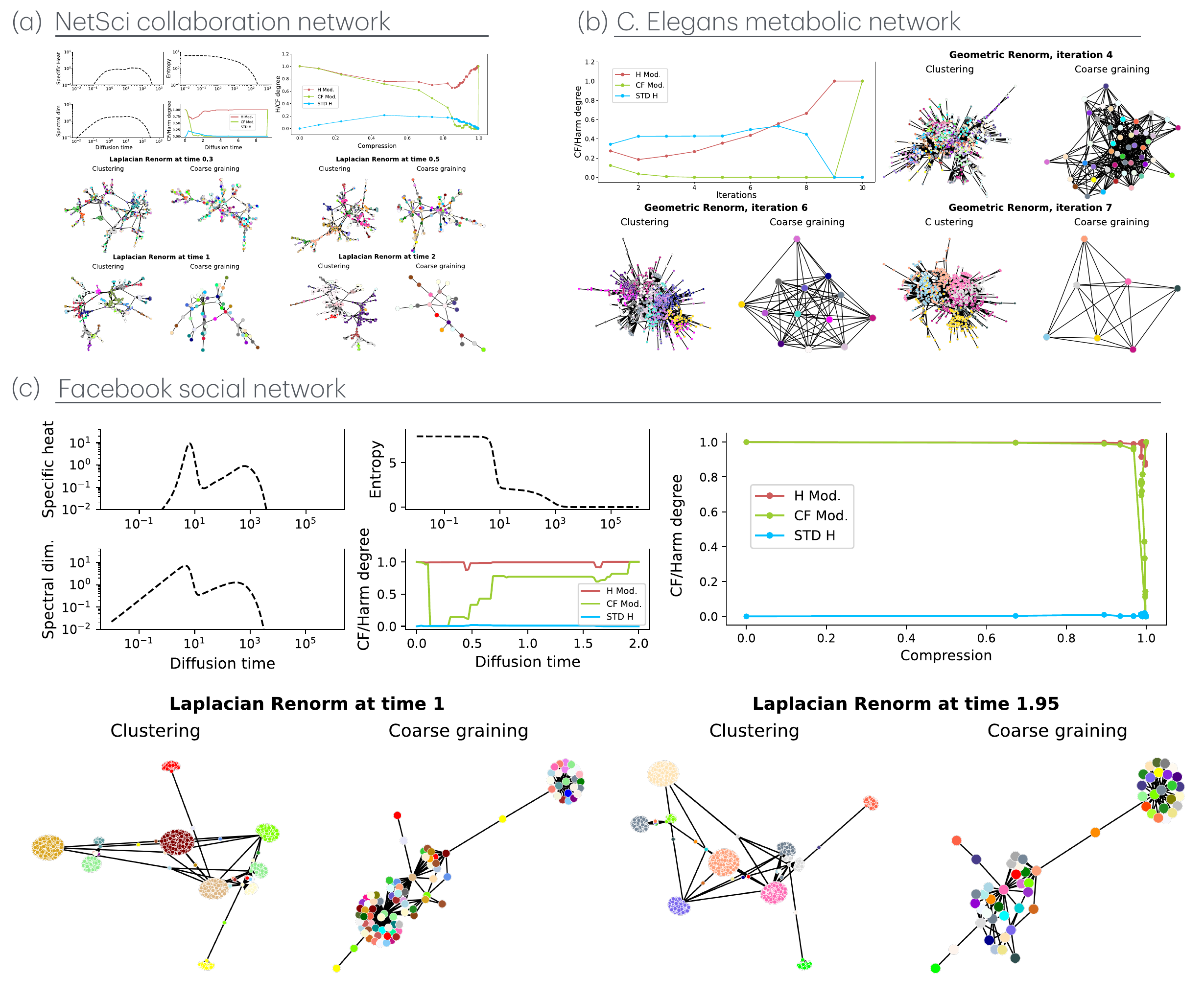}
    \caption{Laplacian renormalization examples for (a) the NetSci collaboration network, (b) the \textit{C.~elegans} metabolic network, and (c) the Facebook social network. 
    In the NetSci network, harmonicity is high at small scales, where clusters are linked through one or few interface nodes ($t=0.3$, $H_{\mathrm{mod}}=0.87$), decreases at intermediate scales due to asymmetric merging ($t=0.5$, $H_{\mathrm{mod}}=0.71$; $t=1$, $H_{\mathrm{mod}}=0.65$), and rises again at high compression through local collapses ($t=2.0$, $H_{\mathrm{mod}}=0.80$). 
    In the \textit{C.~elegans} metabolic network, the relatively homogeneous degree distribution produces a sharper high--low--high pattern: only strongly connected pathways merge at small scales, while at high compression a single dominant macro-node absorbs most of the network. 
    In the Facebook network, bridge nodes preserve harmonicity throughout the RG flow, and the coarse-graining at $t=1.0$ is an exact harmonic morphism; at $t=1.95$ one finds $H_{\mathrm{mod}}=0.99$. For the harmonic curves, $\varepsilon=10^{-3}$ is used in panels (a) and (c), while the single coarse-graining plots are computed with $\varepsilon=0$; in panel (b), $\varepsilon=0$ throughout.}
    \label{fig:si:laplacian_renorm_examples}
\end{figure}
Laplacian renormalization produces the richest variety of harmonic degree behaviors across individual networks, all variations on the high-low-high theme identified in the main text. We present three networks that span this range, together with the Facebook network discussed in the main text.
 
The NetSci collaboration network (Fig.~\ref{fig:si:laplacian_renorm_examples}a) exemplifies the canonical high-low-high pattern with a pronounced intermediate dip. At diffusion time $t = 0.3$ ($H_{\text{mod}} = 0.87$), the coarse-graining consists of small, tightly connected research groups linked to the rest of the network through one or two interface nodes---a topology that naturally satisfies horizontal conformality. As diffusion time increases to $t \simeq 0.5$--$1.0$, the expanding diffusion horizon causes groups to merge asymmetrically: some research communities consolidate faster than others, and the resulting macro-sets develop unbalanced boundary multiplicities ($H_{\text{mod}}$ drops to $0.67$). At very high compression ($t = 2.5$, $H_{\text{mod}} = 0.94$), the network has coalesced into a few large macro-sets, and the coarse-graining increasingly resembles a series of local collapses---peripheral nodes merging into hubs---which are inherently harmonic.

The C.~Elegans metabolic network (Fig.~\ref{fig:si:laplacian_renorm_examples}b) shows a more abrupt version of the high-low-high pattern. At small diffusion times, Laplacian renormalization groups only the most strongly connected metabolic pathways, preserving the random walk structure of the broader network. However, the C.~Elegans metabolic network has a relatively homogeneous degree distribution compared to social or collaboration networks, so diffusion spreads more uniformly as $t$ increases. This produces a rapid transition from many small, well-separated macro-sets to a single dominant macro-node absorbing most of the network, with a corresponding sharp drop in $H_{\text{mod}}$. The moderately high harmonic values even at high compression reflect a general property: when most nodes have been absorbed into a single macro-set, the few remaining boundary nodes are necessarily few relative to the interior, limiting the potential for asymmetric multiplicities.

The Facebook social network (Fig.~\ref{fig:si:laplacian_renorm_examples}c) represents the opposite extreme: $H_{\text{mod}} \approx 1.0$ across the entire renormalization flow, with an exact harmonic morphism at $t = 1.0$. The visualization reveals the mechanism discussed in the main text. Large groups of nodes connected to hubs collapse into macro-nodes, while individual bridge nodes---those with balanced external connectivity---map to themselves. The coexistence of these two populations produces macro-sets with constant boundary multiplicities. The bridge nodes are crucial: they sit at the interfaces between communities and distribute their connections uniformly across neighboring macro-sets, ensuring that the harmonic morphism condition holds globally even as the underlying partition changes with diffusion time.
 
The comparison across networks suggests several patterns. 
The depth of the intermediate minimum in the Laplacian curve appears to correlate with network heterogeneity: networks with strong community structure and heterogeneous degree distributions (\textit{Facebook}, \textit{Web-edu}) show shallow or absent minima because their hub-and-bridge topology better maintains balanced boundary structure. 
Networks with more homogeneous connectivity (\textit{Euroroad}, \textit{C.~Elegans}) show deeper minima because diffusion spreads more uniformly, merging communities at comparable rates and creating competition for boundary nodes. 
Collaboration networks (\textit{NetSci}, \textit{CS Collab}) fall between these extremes, with moderate minima reflecting their intermediate degree heterogeneity.

\subsection{GNN-based renormalization: individual networks}
 
GNN-based renormalization exhibits the least variation across networks of the three methods, consistent with its uniformly low harmonicity signature. Figure~\ref{fig:si:gnn_renormalization_examples} shows 25 independent samples from the soft assignment distribution for the \textit{Weaver} animal social network. 
Two features are characteristic. 
First, the harmonic degree is low across all samples and all compression levels, confirming that the spectrum-preserving objective is systematically incompatible with random walk preservation. Second, there is non-negligible variance between samples at the same nominal compression, reflecting the stochasticity of the soft-to-hard assignment conversion. 
This variance however does not qualitatively alter the fingerprint: no individual sample achieves $H_{\text{mod}}$ values comparable to those from Laplacian renormalization.
 
The near-equality $H_{\text{mod}} \approx CF_{\text{mod}}$ observed in the main text for the \textit{Euro-Road} network persists across networks, including the \textit{Weaver} network. This indicates that whatever harmonicity exists in GNN partitions is driven almost entirely by singleton macro-sets (nodes mapped to themselves), which trivially satisfy both horizontal conformality and combinatorial conformality. 
In contrast, multi-node macro-sets produced by the GNN almost never exhibit balanced boundary multiplicities, because the method distributes nodes based on spectral role similarity rather than local adjacency structure.
 
\begin{figure}[tbp]
    \centering
    \includegraphics[width=0.8\linewidth]{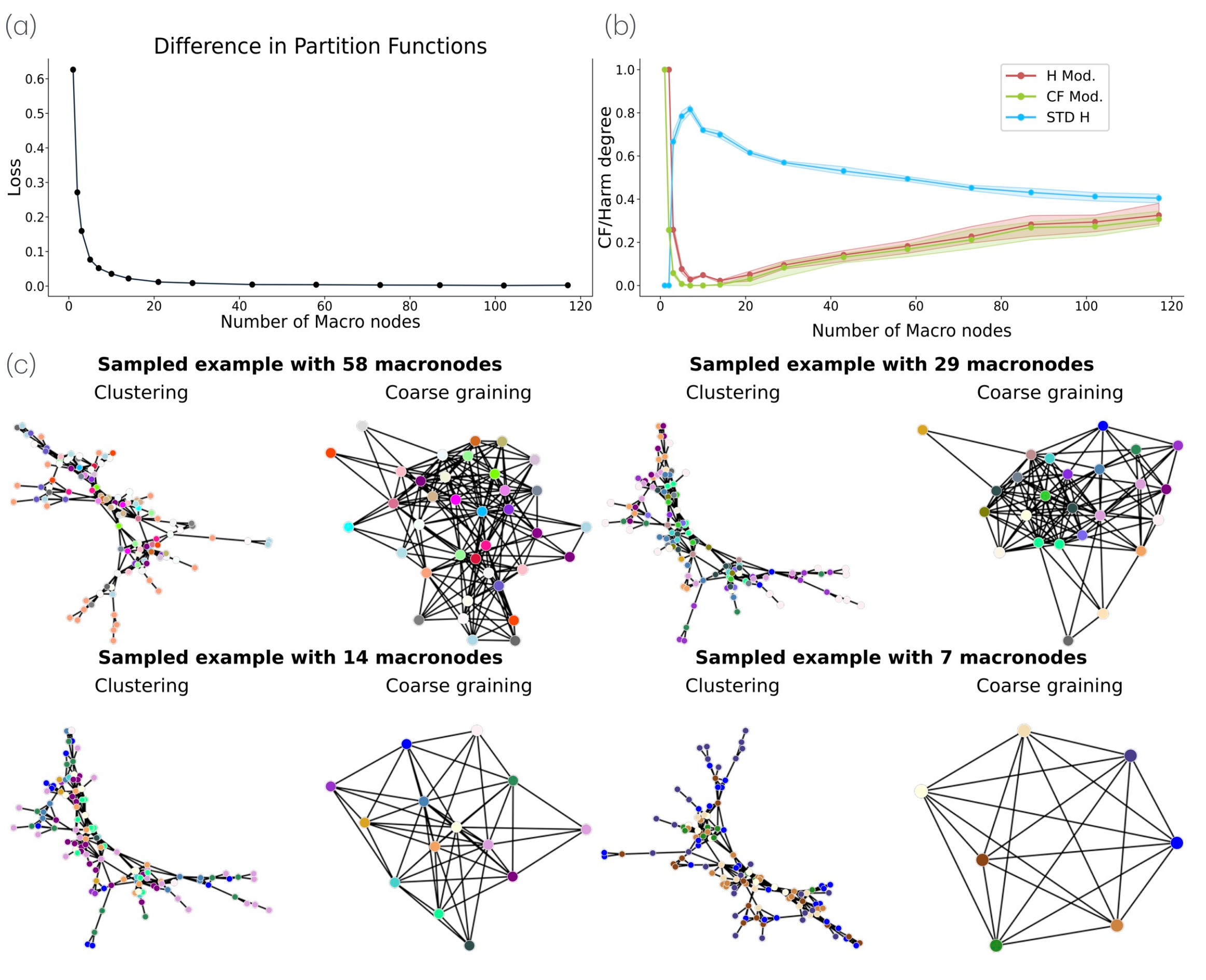}
    \caption{GNN-based renormalization of the \textit{Weaver} social network. 
    (a) Difference in partition functions as a function of the number of macronodes, showing that the sampled partitions rapidly converge toward the same coarse-graining statistics as resolution increases. 
    (b) Harmonicity measures across scales, including modified harmonic degree, cluster-factor modularity, and standard harmonic degree. 
    (c) Representative sampled hard partitions and their associated coarse-grained networks at four compression levels (58, 29, 14, and 7 macronodes). Consistent with the original analysis, harmonicity remains uniformly low across samples, with noticeable inter-sample variability at fixed compression reflecting the stochasticity of the soft-to-hard conversion.}
    \label{fig:si:gnn_renormalization_examples}
\end{figure}

%% ====================================================================
\section{Equilibrium Laplacian renormalization}
\label{sec:eq_laplacian_sm}

Equilibrium Laplacian Renormalization~\cite{yi2025equilibriumpreservinglaplacianrenormalizationgroup} modifies the standard approach to better preserve small eigenmodes and approximate the partition function. The algorithm fixes compression rate $r$, determines the appropriate time $t_r$ via $S(t_r) = \ln((1-r\kappa)N)$, then sorts all node pairs by $\varrho_{ij}\varrho_{ji}/(\varrho_{ii}\varrho_{jj})$ (where $\varrho = e^{-tL}$ is the unnormalized heat kernel) and merges the top pairs until reaching compression $(1-r)N$.

We observe a striking reversal of standard Laplacian behavior: harmonic degree \emph{decreases} at low scales and \emph{increases} at high scales. 
This reversal reflects the global nature of the sorting procedure. Unlike standard Laplacian renormalization (which merges whenever the local threshold is exceeded), Equilibrium Laplacian ranks all pairs globally, then merges top-ranked pairs. This top-down procedure can destroy harmonicity at low scales, but then recovers a giant component toward meso/high scales, which is naturally biased toward high harmonicity.
The conformality is remarkably high ($CF_{\text{mod}} > 0.6$ across most scales), exceeding all other renormalization methods. This occurs because many nodes, left unmerged by the global sorting procedure, are mapped to themselves, trivially satisfying conformality through local identity mapping.

\begin{figure}[tbp]
    \centering
    \includegraphics[width=.8\linewidth]{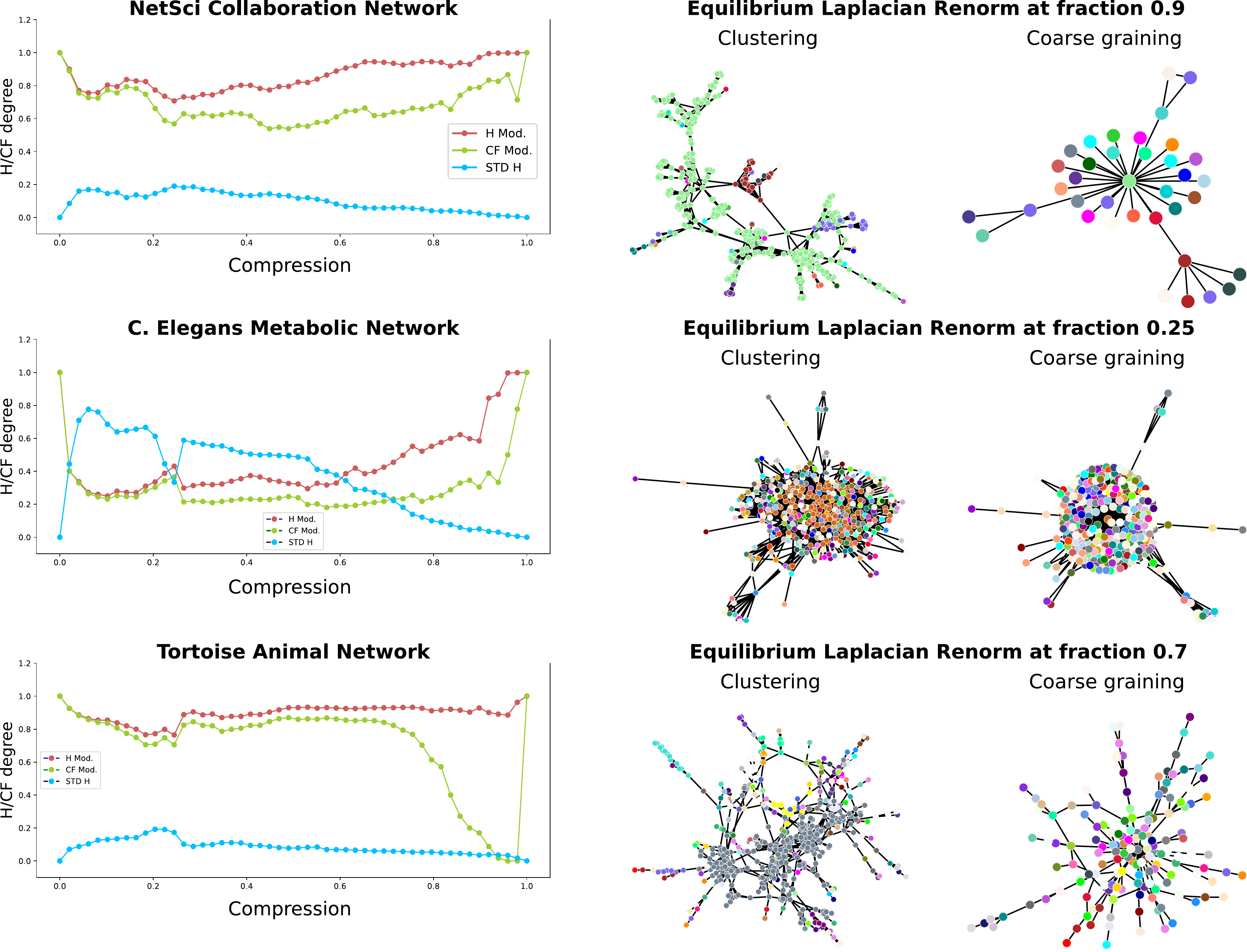}
    \caption{Equilibrium Laplacian Renormalization for several networks: NetSci Collaboration Network, C.~Elegans Metabolic network, and Tortoise Animal Network. At low scales there is a decrease in harmonic degrees, and overall the values are very high. Equilibrium Laplacian Renormalization tends to produce giant components at high scales.}
    \label{fig:si:eq_lap}
\end{figure}

%% ====================================================================
\section{Higher-order extensions}
\label{sec:higher_order_sm}

\subsection{Generalized renormalization scheme for higher-order operators}

To extend Laplacian renormalization to higher-order operators (Hodge, Bochner, $XO$-Laplacians), we introduce a generalized merging procedure that handles operators with non-trivial kernels. Given a Laplacian $L$ on a simplicial complex, its exponential $\varrho = e^{-t L}$ is the heat kernel driving diffusion. The standard merging criterion $\rho_{ij} \geq \min(\rho_{ii}, \rho_{jj})$ is valid when the kernel is spanned by the all-ones vector, as for the graph Laplacian. For generic operators with non-trivial kernels (e.g. Hodge Laplacian of simplicial complexes with non-trivial topology \cite{Muhammad2006Control_Higher_order_Laplacian}), the persistent kernel component would distort the merging criterion.

Our generalized procedure is as follows. For each basis vector $\mathbf{e}_j$: 
\begin{enumerate}
    \item decompose $\mathbf{e}_j = \mathbf{e}_j^{(h)} + \mathbf{e}_j^{(r)}$ into harmonic and residual components by projecting onto the orthogonal complement of the kernel; 
    \item diffuse the residual part: $\varrho_{\cdot,j} = e^{-t L} \mathbf{e}_j^{(r)}$; 
    \item merge nodes $i$ and $j$ if $(|\varrho_{ij}| + |\varrho_{ji}|)/2 \geq \min(|\varrho_{ii}|, |\varrho_{jj}|)$, where absolute values account for possible sign changes in non-symmetric operators.

\end{enumerate}

This procedure commutes with the heat kernel (since $L$ and $e^{-t L}$ share eigenvectors), and reduces to standard Laplacian renormalization for the graph Laplacian and cross-order Laplacians, whose kernels are spanned by the all-ones vector. 
For the Hodge Laplacian, the kernel corresponds to $k$-dimensional topological holes; for the Bochner Laplacian, to balanced configurations of the signed parallel adjacency graph.

\subsection{Node-based metrics for higher-order renormalization}

As a first approach, we evaluated the node-based harmonic degree on 1-skeletons under higher-order renormalization schemes, using the framework of Ref.~\cite{nurisso2025higher} extended to Hodge diffusion. For each network, we consider the standard Laplacian, the multi-order Laplacian~\cite{Lucas_2020}, and the $XO(1,2)$, $XO(2,1)$, 2-Bochner, and 2-Hodge Laplacians. 

Higher-order Laplacians that diffuse on edges or triangles consistently disrupt node-level harmonicity, because the partition that simplices induce on nodes does not respect node-level diffusive locality. By contrast, the multi-order Laplacian, which remains node-based while incorporating higher-order information, achieves performance comparable to standard Laplacian renormalization. This confirms that the incompatibility is between the diffusion domain ($k$-simplices) and the evaluation domain (nodes), not between higher-order structure and random walk preservation.

\begin{figure}[tbp]
    \centering
    \includegraphics[width=\linewidth]{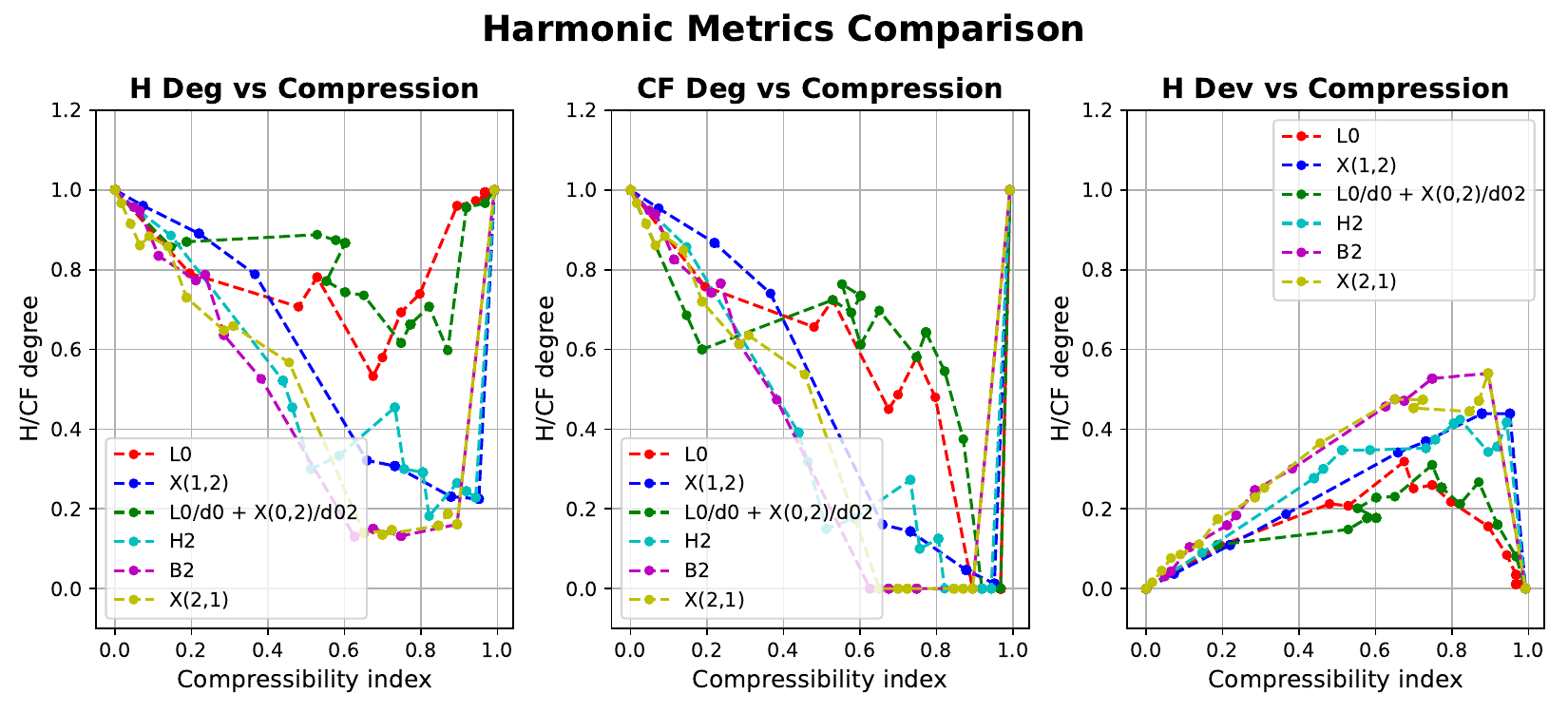}
    \caption{Evaluation of higher-order and node-based renormalization schemes via harmonic degree on the 1-skeletons of pseudofractal complex of dimension 2, built with 4 steps.}
    \label{fig:si:ho_pf2}
\end{figure}

\subsection{Harmonic evaluation on the $k$-order adjacency graph}

To resolve the incompatibility between node-based harmonic degree and higher-order diffusion, we evaluate harmonicity on the graph where the dynamics naturally live: the $k$-order adjacency graph of the simplicial complex. In this graph, vertices represent $k$-simplices and edges connect $k$-simplices sharing a $(k-1)$-face. We then apply the harmonic degree framework of Sec.~\ref*{section:harmonic_degree} to the coarse-graining of this adjacency graph induced by higher-order diffusion.

This approach evaluates the harmonicity of coarse-grainings of the adjacency graph induced by higher-order dynamics, rather than directly assessing coarse-grainings between simplicial complexes. The coarse-grained graphs are pseudo-adjacency graphs: they are not guaranteed to be $k$-skeletons of any simplicial complex.

For Hodge and Bochner diffusion, one can also consider the parallel adjacency graph, defined via the parallel adjacency relation following the Bochner decomposition~\cite{Forman2003BochnersMF}. Since it is unclear which structure---adjacency or parallel adjacency---should be preferred as the ambient space for these dynamics, we analyze both. We note that when the dimension of the simplices is maximal (i.e., equal to the dimension of the complex), the two graphs coincide. This is the case for Fig.~\ref{fig:higher_order}, where the 2-dimensional pseudofractal has no tetrahedra. For the 3-dimensional case (Fig.~\ref{fig:si:pf3d}), the notions do not coincide for 2-diffusion but do for 3-diffusion.

The key finding is that the Hodge Laplacian is sharply sensitive to the topological dimension of the complex:
\begin{itemize}
\item On a 2-dimensional pseudofractal, 2-diffusion via the Hodge Laplacian produces persistent harmonic morphisms (Fig.~\ref{fig:higher_order},~\ref{fig:si:pf2d}), consistently outperforming the Bochner and $XO$ Laplacians. This advantage is attributable to the Forman curvature term~\cite{Forman2003BochnersMF}, the reactive component of Hodge diffusion~\cite{nurisso2025higher}.
\item On a 3-dimensional pseudofractal, 2-diffusion via the Hodge Laplacian completely disrupts harmonicity regardless of whether adjacency or parallel adjacency is used, while 3-diffusion recovers high harmonic degree values (Fig.~\ref{fig:si:pf3d}).
\end{itemize}

\begin{figure}[tbp]
    \centering
    \includegraphics[width= 1\linewidth]{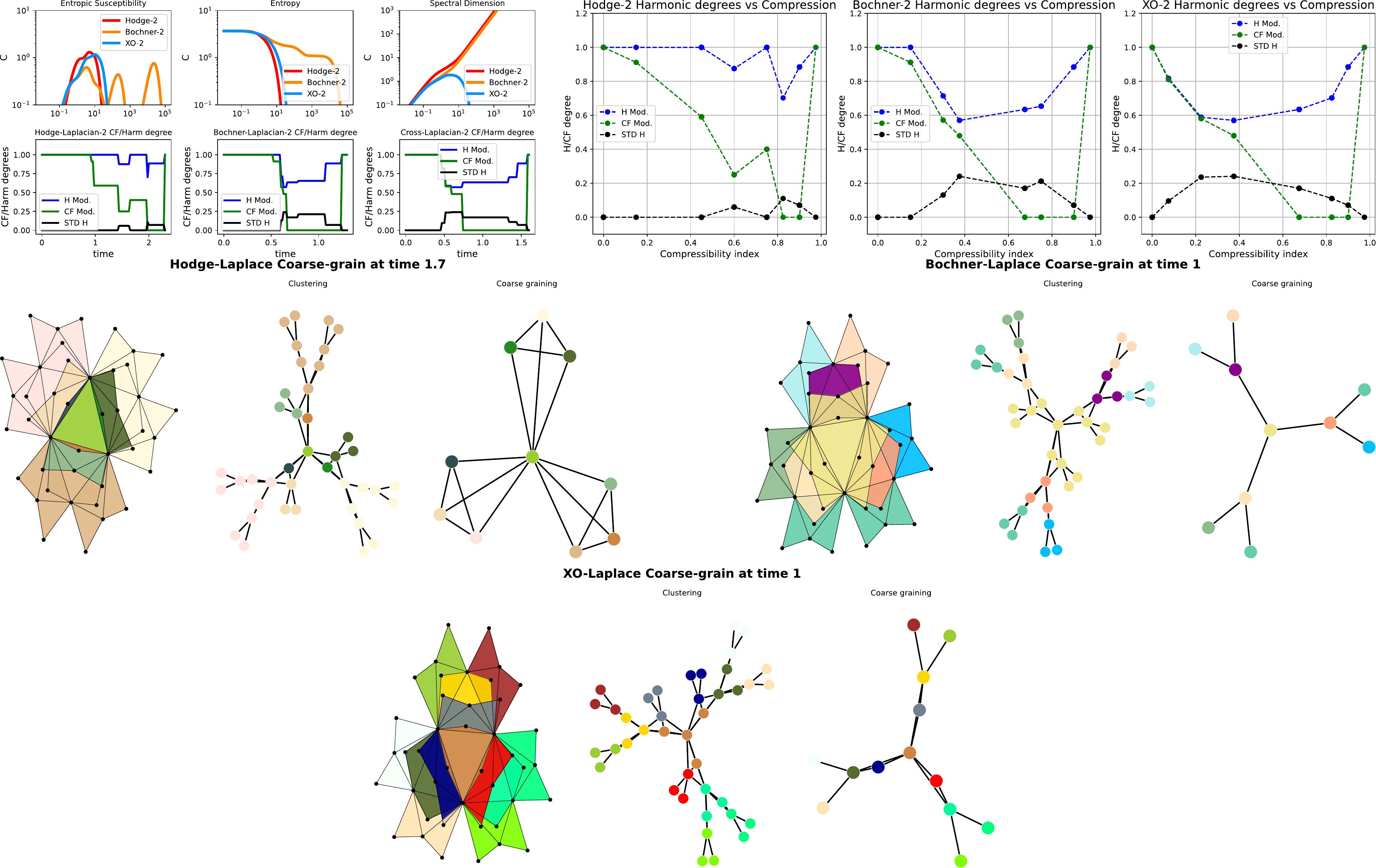}
    \caption{Harmonic degrees of Pseudofractal simplicial complex of dimension $2$. We highlight that in this case there are no $2$-D holes and adjacency coincides with parallel adjacency, since there are no tetrahedra. We can see that the Hodge-Laplacian is the one that preserves most the harmonicity. The Hodge-Laplace transformation shown is indeed a harmonic morphism.}
    \label{fig:si:pf2d}
\end{figure}

\begin{figure}[tbp]
    \centering
    \includegraphics[width=1 \linewidth]{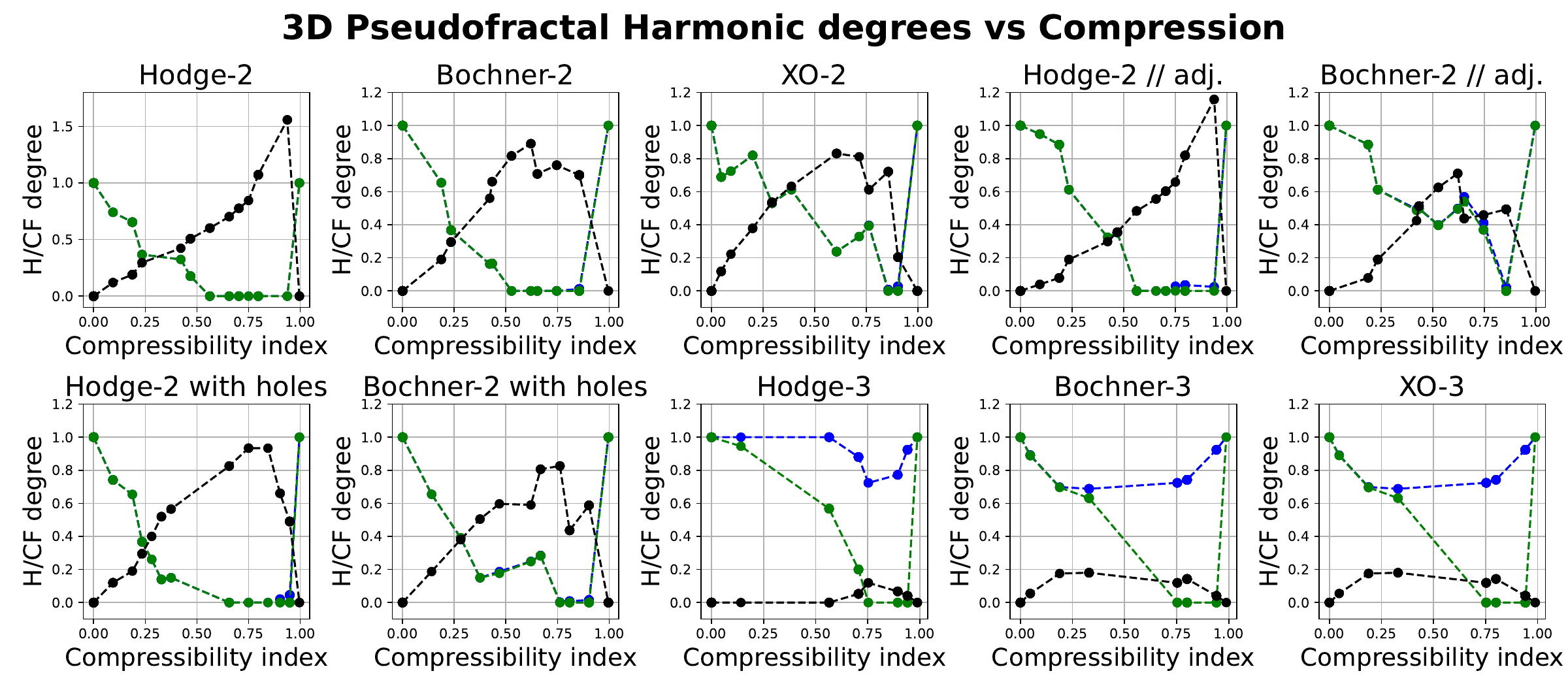}
    \caption{Harmonic degrees vs compression of second and third order operators on a pseudofractal complex of dimension $3$. The 2-order Hodge operator completely disrupts harmonicity in all cases (with or without holes, adjacency or parallel adjacency), while the 3-order Hodge operator yields the highest harmonic degree values. Blue: modified harmonic degree; green: modified conformal degree; black: harmonic deviation.}
    \label{fig:si:pf3d}
\end{figure}

\subsection{Real-world higher-order networks}

We analyze several real-world higher-order networks previously studied in Ref.~\cite{Iacopini2019}: the social/contact networks of a hospital (LH10), high-school (Thiers13), and workplace (InVS15). For each network, we consider only the largest connected component of 2-simplices and analyze it using 2-order diffusion.

For LH10 and InVS15, harmonicities are completely disrupted, as in the 3-dimensional pseudofractal case, suggesting that these networks require higher-order diffusion to capture their full structure. By contrast, Thiers13 shows overall moderate-to-high harmonicity values, with the pattern of high harmonicity at small scales and moderate decrease at intermediate scales observed in graph Laplacian cases (Fig.~\ref{fig:si:thiers13}). This suggests that the Thiers13 network can be meaningfully described in terms of its 2-simplex (triplet) interactions.

\begin{figure}[tbp]
    \centering
    \includegraphics[width=1\linewidth]{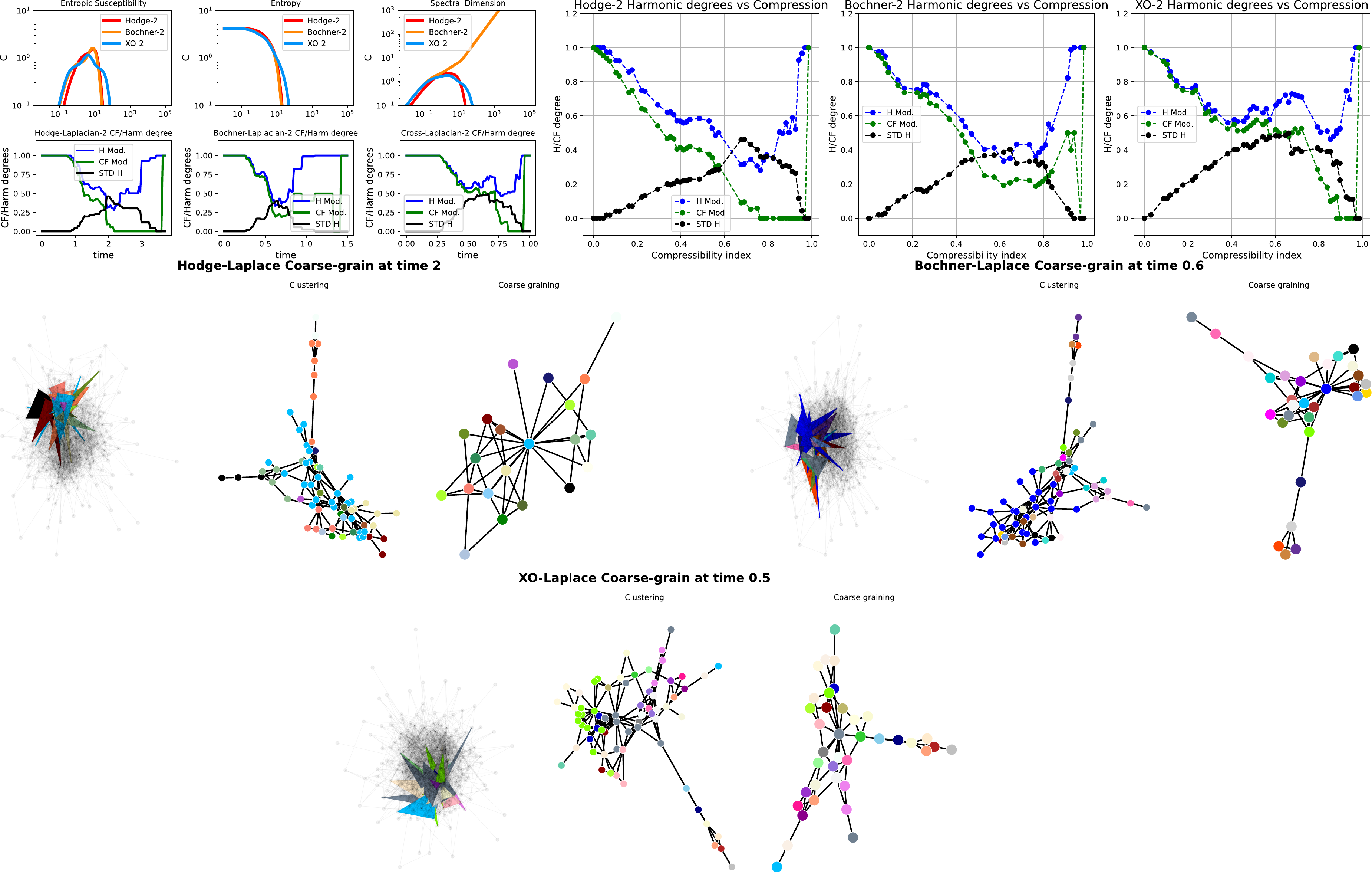}
    \caption{Harmonic degrees of the Thiers13 simplicial complex under 2-order operators. At small scales harmonicity is moderately high, with a decrease at intermediate scales. The $XO$-Laplacian preserves harmonicity best in this case.}
    \label{fig:si:thiers13}
\end{figure}

\subsection{Higher-order harmonic morphisms}

Let $A$ be an operator over the space of functions defined on the collection $S_k$ of $k$-simplices of a simplicial complex. A function is locally $A$-harmonic at simplex $\sigma$ if $Af(\sigma) = 0$.

Given simplicial complexes $S$ and $\mathcal{S}$ with respective $k$-simplex collections $S_k$ and $\mathcal{S}_k$, a function $\varphi_k: S_k \to \mathcal{S}_k$ is a \emph{Hodge-harmonic morphism of degree $k$} if: for each $\sigma' \in \mathcal{S}_k$ and for each function $f$ that is locally Hodge-harmonic at $\sigma'$, the composition $f \circ \varphi_k$ is locally Hodge-harmonic at $\sigma$ for all $\sigma \in \varphi_k^{-1}(\sigma')$. Analogously, one defines \emph{Bochner-harmonic morphisms} by replacing Hodge-harmonicity with Bochner-harmonicity.

Concrete examples of Hodge-harmonic morphisms include the surjection from an infinite 4-colored tessellation of 2-simplices arranged on a honeycomb lattice onto an empty tetrahedron, and the merging of opposite faces of an octahedron into an empty tetrahedron.

However, both Hodge-harmonic and Bochner-harmonic morphisms are orientation-dependent, and it is unclear whether they admit a higher-order analog of horizontal conformality. Since orientations do not arise naturally in our data, we do not pursue these definitions empirically and leave deeper investigation to future work. We note that cross-order harmonic morphisms, by contrast, coincide with standard harmonic morphisms of the corresponding $k$-order adjacency graphs.

%% ====================================================================
\section{Harmonic morphism optimization}
\label{sec:optimization_sm}

Finding partitions that are harmonic morphisms on graphs is a completely unexplored question, related to graph gonality~\cite{Caporaso_Gonality_2}---the existence of harmonic morphisms from $G$ to trees connects to deep questions in algebraic geometry.

We formulate three optimization problems:
\begin{align}
&\max_{p} H_{\text{mean}}(p), \\
&\max_{p} H_{\text{mod}}(p), \\
&\max_{p} (-H_{\text{Dev}}(p)),
\end{align}
where $p$ is a partition of $V$. For meaningful results, we require $|p| \geq 3$; with $|p|=2$, every partition is trivially a harmonic morphism.

We introduce \textbf{harmonic-corrected modularity}:
\begin{equation}
\max_p \left[ \text{MOD}(p) + \lambda \cdot HM(p) \right],
\end{equation}
where $\text{MOD}(p)$ is Newman's modularity~\cite{Modularity_Newman_2006}, $HM(p) = -H_{\text{Dev}}$, and $\lambda > 0$ balances assortativity and harmonicity. The rationale: greedy modularity often produces high harmonic values (Section~\ref{sec:clustering_sm}), suggesting modularity-driven partitions are partly compatible with random walk structure. The harmonic regularization term refines these partitions toward better dynamical preservation.

We implement a greedy algorithm (code available on the code repository). The algorithm performs well on small graphs (up to $\sim 60$ nodes) but encounters local optima on larger networks. We also attempted machine learning approaches using the matrix criterion from Ref.~\cite{Melles2024} and adapting the GNN framework~\cite{Coarse_graining_network_De_Domenico} with loss $\|L_1\Phi - \Phi L_2\|^2$, where $\Phi \in \mathbb{R}^{|V| \times |\mathcal{V}|}$ is the grouping matrix. When the loss vanishes, $\Phi L_2 = L_1 \Phi$, implying $\varphi$ is harmonic. However, these approaches performed poorly, likely because soft-to-hard assignment conversion disrupts optimization.

%% ====================================================================
\section{Connection to conformal field theory}
\label{si:link-to-conformal-field-theory}

In two-dimensional conformal field theory~\cite{Introduction_to_Conformal_Field_Theory}, conformal invariance constrains correlation functions and operator algebras, leading to exact solutions of critical systems. For discrete lattice models, establishing conformal invariance in the scaling limit~\cite{Conformal_Invariance_Lattice_Models_Smirnov,Discrete_Complex_Analysis} has been a major achievement, explaining universal properties of critical phenomena. 

Our framework suggests that harmonic morphisms provide analogous structure for networks: just as conformal maps connect different coordinate systems for the same physical spacetime, harmonic morphisms connect different resolutions of the same underlying diffusion process. It is worth noting that the standard block-grid transformations used in the Kadanoff renormalization group do not constitute harmonic morphisms. Nevertheless, given the deep connections between harmonic morphisms and complex algebraic geometry, these transformations are natural candidates for further developing the ideas of discrete complex analysis~\cite{Discrete_Complex_Analysis}.

Importantly, our definitions are purely combinatorial, avoiding edge weights or geometric embeddings~\cite{conformallycovariantoperators,Bobenko_Conformal_2015,Conformal_Squid}. This abstraction comes at a cost---we cannot visualize angle preservation---but provides a benefit: the framework applies to any network topology. The price of generality is that we measure conformality through multiplicities (counts of neighbors in each macro-set) rather than geometric quantities.

The analogy extends to renormalization. In lattice field theories, RG transformations preserving conformal symmetry identify fixed points and universal scaling behavior. For networks, Theorem~\ref*{theo:RW_preservation} suggests harmonic morphisms preserve the ``universal'' property of diffusion---the probabilistic structure governing how random walkers move between macro-sets. Networks admitting sequences of harmonic morphisms across scales would exhibit a form of diffusive self-similarity, analogous to scale invariance in critical phenomena.

\clearpage

\putbib[biblio]
\end{bibunit}

\end{document}